\documentclass{article}
\usepackage{arxiv}

\usepackage{natbib}
\usepackage[utf8]{inputenc} % allow utf-8 input
\usepackage[T1]{fontenc}    % use 8-bit T1 fonts
\usepackage{hyperref}       % hyperlinks
\usepackage{url}            % simple URL typesetting

\usepackage{amssymb}
\usepackage{caption}
\usepackage{subcaption}

\usepackage{algpseudocode}
\usepackage{algorithm}

\usepackage{tabularx}
\usepackage{enumitem}
\usepackage{xpatch}

\usepackage{mdframed}
\usepackage{tabu,booktabs}

\makeatletter
\xpatchcmd{\algorithmic}
  {\ALG@tlm\z@}{\leftmargin\z@\ALG@tlm\z@}
  {}{}
\makeatother
\floatname{algorithm}{Alg$.$}

\newcommand{\ofr}{\textsc{Offr}\xspace}
\newcommand{\static}{\textit{static}\xspace}
\newcommand{\ideal}{\textit{ideal}\xspace}
\newcommand{\dynamic}{\textit{dynamic}\xspace}

\newcommand{\lastfm}{\textit{lastfm15k}\xspace}
\newcommand{\batchfw}{\textit{batch-FW}\xspace}
\newcommand{\fairco}{\textit{FairCo}\xspace}
\newcommand{\mlm}{\textit{MovieLens-1m}\xspace}

\newcommand{\muij}{\mu_{ij}}

\newcommand{\mui}{\mu_i}
\newcommand{\muit}{\mu_{\it}}
\newcommand{\sigmat}{\sigma^{(t)}}

\newcommand{\K}{k}
\renewcommand{\k}{r}

\newcommand{\wei}{b}

\newcommand{\ui}{u_i}
\newcommand{\uh}{\hat{u}}

\newcommand{\uhiz}{\uh^{(0)}_i}
\newcommand{\uhi}{\uh_i}
\newcommand{\uhit}{\uh^{(t)}_i}
\newcommand{\uhitt}{\uh^{(t)}_{\it}}
\newcommand{\uhittmo}{\uh^{(t-1)}_{\it}}

\renewcommand{\v}{v}

\newcommand{\vj}{v_j}
\newcommand{\vjp}{v_{j'}}
\newcommand{\vh}{\hat{v}}
\newcommand{\vht}{\vh^{(t)}}
\newcommand{\vhjz}{\vh^{(0)}_j}
\newcommand{\vhjt}{\vh^{(t)}_j}
\newcommand{\vhj}{\vh_j}

\newcommand{\vhtmo}{\vh^{(t-1)}}

\newcommand{\topk}{\mathrm{top\K}}

\newcommand{\arms}{{\mathcal{E}}}
\newcommand{\A}{E}
\renewcommand{\a}{e}
\newcommand{\simpA}{\overline{\arms}}

\newcommand{\w}{w}
\newcommand{\wi}{w_i}

\newcommand{\pis}{\Pi}
\newcommand{\pii}{\pi_i}
\newcommand{\pione}{\pi_1}
\newcommand{\piimo}{\pi_{i{-}1}}
\newcommand{\piipo}{\pi_{i{+}1}}
\newcommand{\pin}{\pi_{n}}
\newcommand{\pipi}{\pi'_i}

\newcommand{\piopt}{\pi^*}
\newcommand{\piij}{\pi_{ij}}

\newcommand{\quaj}{q_j}
\newcommand{\quajp}{q_{j'}}

\newcommand{\quah}{\hat{q}}
\newcommand{\quahj}{\quah_j}
\newcommand{\quaht}{\quah^{(t)}}
\newcommand{\quahtmo}{\quah^{(t-1)}}
\newcommand{\quahjt}{\quah^{(t)}_j}
\newcommand{\quahjz}{\quah^{(0)}_j}
\newcommand{\quaavg}{q_{\mathrm{avg}}}
\newcommand{\quahavgt}{\quah^{(t)}_{\mathrm{avg}}}
\newcommand{\quahavg}{\quah_{\mathrm{avg}}}

\newcommand{\groups}{\mathcal{S}}
\newcommand{\group}{s}%{\mathfrak{s}}
\newcommand{\groupofi}{\group[i]}
\newcommand{\vgj}{\v_{j|\group}}
\newcommand{\vgij}{\v_{j|\groupofi}}
\newcommand{\vavgj}{\v_{j|\mathrm{avg}}}
\newcommand{\vhgj}{\hat{\v}_{j|\group}}
\newcommand{\vhavgg}{\hat{\v}_{\mathrm{avg}|\group}}
\newcommand{\vhgjt}{\hat{\v}^{(t)}_{j|\group}}
\newcommand{\vhavgjt}{\hat{\v}^{(t)}_{j|\mathrm{avg}}}
\newcommand{\vhgij}{\hat{\v}_{j|\groupofi}}
\newcommand{\vhavgj}{\hat{\v}_{j|\mathrm{avg}}}

\newcommand{\wsumg}{\overline{w}_{\group}}
\newcommand{\wsumgofi}{\overline{w}_{\groupofi}}
\newcommand{\wming}{\overline{w}_{\min}}

\newcommand{\welf}{f}

\newcommand{\pdevwi}{g_{w, i}}
\newcommand{\pdevwij}{g_{w, ij}}

\newcommand{\pdevwhti}{g_{\wht, i}}
\newcommand{\pdevwhtij}{g_{\wht, ij}}
\newcommand{\pdevwhij}{g_{\wh, ij}}

\newcommand{\pdevwhtmoi}{g_{\whtmo, i}}

\newcommand{\boundA}{B_\arms}
\renewcommand{\a}{e}
\newcommand{\at}{\a^{(t)}} % -> \jt, \ejt
\newcommand{\atp}{\a^{(\tp)}} % -> \jtp, \ejtp
\newcommand{\atpj}{\a^{(\tp)}_j} % -> \jtp, \ejtp

\newcommand{\aopt}{\tilde{\a}}
\newcommand{\aoptit}{\aopt^{(t)}_i} % \eoptit
\newcommand{\ahit}{\hat{\a}^{(t)}_i} % \ehit
\newcommand{\aoptitpo}{\aopt^{(t+1)}_i} % \eoptitpo
\newcommand{\ahitpo}{\hat{\a}^{(t+1)}_i} % \ehitpo

\newcommand{\tp}{\tau}
\newcommand{\pih}{\pi}%{\hat{\pi}}
\newcommand{\pihz}{\pih_0}
\newcommand{\piht}{\pih^{(t)}}
\newcommand{\pihtpo}{\pih^{(t+1)}}
\newcommand{\pihtmo}{\pih^{(t-1)}}

\newcommand{\ione}{{i^{(1)}}}
\renewcommand{\it}{{i^{(t)}}}
\newcommand{\itpo}{{i^{(t+1)}}}
\newcommand{\itmo}{{i^{(t-1)}}}
\newcommand{\itp}{{i^{(\tp)}}} 
\newcommand{\ct}{c^{(t)}}
\newcommand{\cti}{c^{(t)}_i}
\newcommand{\ctg}{c^{(t)}_{\group}}
\newcommand{\ci}{c_{i}}
\newcommand{\cg}{c_{\group}}
\newcommand{\cgi}{c_{\groupofi}}

\newcommand{\pihit}{\pih^{(t)}_i}
\newcommand{\pihijt}{\pih^{(t)}_{ij}}
\newcommand{\pihij}{\pih_{ij}}
\newcommand{\pihitpo}{\pih^{(t+1)}_i}

\newcommand{\wh}{\hat{\w}}
\newcommand{\wht}{{\wh^{(t)}}}
\newcommand{\whtmo}{{\wh^{(t-1)}}}
\newcommand{\whit}{{\wh^{(t)}_i}}
\newcommand{\whi}{{\wh_i}}
\newcommand{\whitmo}{{\wh^{(t-1)}_i}}

\newcommand{\obj}{f}
\newcommand{\objw}{\obj_\w}
\newcommand{\objwh}{\obj_{\wh}}

\newcommand{\objwpii}{\obj^{\w, \pi}_i}

\newcommand{\objwmax}{\obj_\w^*}

\newcommand{\nablai}{\nabla_i}

\newcommand{\g}{g}
\newcommand{\git}{\g^{(t)}_i}
\newcommand{\gitpo}{\g^{(t+1)}_i}
\newcommand{\gh}{\hat{g}}
\newcommand{\ghi}{\hat{g}_i}
\newcommand{\ghit}{\gh^{(t)}_i}
\newcommand{\ghij}{\gh_{ij}}
\newcommand{\ghitpo}{\gh^{(t+1)}_i}
\newcommand{\ghitt}{\gh^{(t)}_{\it}}

\newcommand{\lipgw}{L}
\newcommand{\lipgwi}{\lipgw_i} %% currently in use
\newcommand{\boundigw}{G_i}
\newcommand{\boundegi}{D_i}
\newcommand{\curvf}{C}

\newcommand{\curvif}{\curvf_i}
\newcommand{\regt}{R^{(t)}}
\newcommand{\regtpo}{R^{(t+1)}}
\newcommand{\pregt}{\overline{R}^{(t)}}
\newcommand{\pregtpo}{\overline{R}^{(t+1)}}

\newcommand{\calCt}{\mathcal{C}^{(t)}}
\newcommand{\calRt}{\mathcal{R}^{(t)}}
\newcommand{\calDt}{\mathcal{D}^{(t)}}
\newcommand{\calGt}{\mathcal{G}^{(t)}}
\newcommand{\caltGt}{\tilde{\mathcal{G}}^{(t)}}

\newcommand{\hgammaitpo}{\hat{\gamma}^{(t+1)}_i}
\newcommand{\gammatpo}{\gamma^{(t+1)}}

\newcommand{\ha}{{\hat{\a}}}

\newcommand{\pihitpot}{\piht_{\itpo}}
\newcommand{\pihitpotpo}{\pihtpo_{\itpo}}
\newcommand{\gitpotpo}{g^{(t+1)}_{\itpo}}
\newcommand{\ctitpo}{c^{(t)}_\itpo}
\newcommand{\objwpihtitpo}{\obj^{\w,\piht}_{\itpo}\!}

\newcommand{\curvitpof}{C_{\itpo}}
\newcommand{\ahitpotpo}{\hat{\a}^{(t+1)}_{\itpo}}

\newcommand{\gitp}{g^{(\tp)}_i}
\newcommand{\ghitp}{\gh^{(\tp)}_i}

\newcommand{\groupone}{\group^{(1)}}
\newcommand{\groupt}{\group^{(t)}}
\newcommand{\grouptp}{\group^{(\tp)}}

\newcommand{\xg}{x_{\group}}
\newcommand{\xgi}{x_{\groupofi}}

\newcommand{\jp}{j'}
\newcommand{\vhjp}{\vh_{\jp}}
\newcommand{\quahjp}{\quah_{\jp}}

\input{generic_aliases}

\newtheorem{theorem}{Theorem}
\newtheorem{proposition}[theorem]{Proposition}
\newtheorem{lemma}[theorem]{Lemma}
\newtheorem{corollary}[theorem]{Corollary}

\title{Fast online ranking with fairness of exposure}

\author{
    Nicolas Usunier, Virginie Do, Elvis Dohmatob \\
    Meta AI\\
    \texttt{\{usunier,virginiedo,dohmatob\}@fb.com}
}

\begin{document}

\maketitle

%%
%% The "title" command has an optional parameter,
%% allowing the author to define a "short title" to be used in page headers.

%%
%% The "author" command and its associated commands are used to define
%% the authors and their affiliations.
%% Of note is the shared affiliation of the first two authors, and the
%% "authornote" and "authornotemark" commands
%% used to denote shared contribution to the research.

% \author{Nicolas Usunier}
% \email{usunier@fb.com}
% \author{Virginie Do}
% \email{virginiedo@fb.com}
% \author{Elvis Dohmatob}
% \email{dohmatob@fb.com}
% \affiliation{%
%   \institution{Meta AI}
% %   \streetaddress{we probably do not care}
%   \city{Paris}
% %   \state{}
%   \country{France}
% %   \postcode{don't care either}
% }

%%
%% By default, the full list of authors will be used in the page
%% headers. Often, this list is too long, and will overlap
%% other information printed in the page headers. This command allows
%% the author to define a more concise list
%% of authors' names for this purpose.
% \renewcommand{\shortauthors}{Usunier et al.}

% \newcommand{\nico}[1]{{\color{blue} #1}}
% \newcommand{\virginie}[1]{{\color{violet} V: #1}}

% \definecolor{midgreen}{rgb}{0.1,0.5,0.1}
% \newcommand{\elvis}[1]{{\color{midgreen} [Elvis: {#1}]}}

\begin{abstract}

As recommender systems become increasingly central for sorting and prioritizing the content available online, they have a growing impact on the opportunities or revenue of their items producers. For instance, they influence which recruiter a resume is recommended to, or to whom and how much a music track, video or news article is being exposed. This calls for recommendation approaches that not only maximize (a proxy of) user satisfaction, but also consider some notion of fairness in the exposure of items or groups of items. Formally, such recommendations are usually obtained by maximizing a concave objective function in the space of randomized rankings. When the total exposure of an item is defined as the sum of its exposure over users, the optimal rankings of every users become coupled, which makes the optimization process challenging. Existing approaches to find these rankings either solve the global optimization problem in a \textit{batch} setting, i.e., for all users at once, which makes them inapplicable at scale, or are based on heuristics that have weak theoretical guarantees. In this paper, we propose the first efficient \textit{online} algorithm to optimize concave objective functions in the space of rankings which applies to every concave and smooth objective function, such as the ones found for fairness of exposure. Based on online variants of the Frank-Wolfe algorithm, we show that our algorithm is \textit{computationally fast}, generating rankings on-the-fly with computation cost dominated by the sort operation, \textit{memory efficient}, and has \textit{strong theoretical guarantees}.
%, with a regret that decreases as $O(1/\sqrt{t})$ where $t$ is the number of time steps. In other words, 
Compared to baseline policies that only maximize user-side performance, our algorithm allows to incorporate complex fairness of exposure criteria in the recommendations %policies 
with negligible computational overhead. We present experiments on artificial music and movie recommendation tasks using Last.fm and MovieLens datasets which suggest that in practice, the algorithm rapidly reaches good performances on three different objectives representing different fairness of exposure criteria.

\end{abstract}

\keywords{fairness, recommender systems, online ranking}

%%
%% This command processes the author and affiliation and title
%% information and builds the first part of the formatted document.

%% Nico: put brackets around inputs so that each section can define 
%% its own newcommands without affecting other sections
{
\section{Introduction}

Recommender systems are ubiquitous in our lives, from the prioritization of content in news feeds to matching algorithms for dating or hiring.
%, to the recommendations of cultural content such as books, music, or movies on dedicated platforms. 
%
The objective of recommender systems is traditionally formulated as maximizing a proxy for user satisfaction such as ranking performance. However, it has been observed that these recommendation strategies can have undesirable side effects. For instance, several authors discussed popularity biases and winner-take-all effects that may lead to disproportionately expose a few items even if they are assessed as only slightly better than others \citep{abdollahpouri2019unfairness,singh2018fairness,biega2018equity}, or disparities in content recommendation across social groups defined  by sensitive attributes \citep{sweeney2013discrimination,imana2021auditing}. An approach to mitigate these undesirable effects is to take a more general perspective to the objective of recommendation systems. Considering recommendation as an allocation problem \citep{singh2018fairness,patro2020fairrec} in which the ``resource'' is the exposure to users, the objective of recommender systems is to allocate this resource fairly, i.e., by taking into account the interests of the various stakeholders -- users, content producers, social groups defined by sensitive attributes -- depending on the application context. This perspective yields the traditional objective of recommendation when only the ranking performance averaged over individual users is taken into account.

There are two main challenges associated with the fair allocation of exposure in recommender systems. The first challenge is the specification of the formal objective function that defines the trade-off between the possibly competing interests of the stakeholders in a given context. The second challenge is the design of a scalable algorithmic solution: when considering the exposure of items across users in the objective function, the system needs to account for what was previously recommended (and, potentially, to whom) when generating the recommendations for a user. This requires solving a global optimization problem in the space of the rankings of all users. In contrast, traditional recommender systems simply sort items by estimated relevance to the user, irrespective of what was recommended to other users.

In this paper, we address the algorithmic challenge, with a solution that is sufficiently general to capture many objective functions for ranking with fairness of exposure, leaving the choice of the exact objective function to the practitioner. Following previous work on fairness of exposure, we consider objective functions that are concave functions that should be optimized in the space of randomized rankings \citep{singh2018fairness,singh2019policy,morik2020controlling,do2021two}. Our algorithm, \ofr (\texttt{O}nline \texttt{F}rank-Wolfe for \texttt{F}air \texttt{R}anking), is a computationally efficient algorithm that optimizes such objective functions \textit{online}, i.e., by generating rankings on-the-fly as users request recommendations. The algorithm dynamically modifies item scores to optimize for both user utility and the selected fairness of exposure objective. We prove that the objective function converges to the optimum in $O(1/\sqrt{t})$, where $t$ is the number of time steps. The computational complexity of \ofr at each time step is dominated by the cost of sorting, and it requires only $O(\#users +\#items)$ storage. The computation cost of \ofr are thus of the same order as what is required in traditional recommenders systems. Consequently, using \ofr, taking into account fairness of exposure in the recommendations is (almost) free.
Our main technical insight is to observe that in the context of fair ranking, the usage of Frank-Wolfe algorithms \citep{frank1956algorithm} resolves two difficulties:
%at the same time:
\begin{enumerate}[leftmargin=*]
    \item Frank-Wolfe algorithms optimize in the space of probability distributions but use at each round a deterministic outcome as the update direction. In our case, it means that \ofr outputs a (deterministic) ranking at each time step while implicitly optimizing in the space of randomized rankings.
    \item Even though the space of rankings is combinatorial, the objective functions used in fairness of exposure have a linear structure that Frank-Wolfe algorithms can leverage, as already noticed by \citet{do2021two}.
\end{enumerate}

Compared to existing algorithms, \ofr is the first widely applicable and scalable algorithm for fairness of exposure in rankings. Existing online ranking algorithms for fairness of exposure \citep{morik2020controlling,biega2018equity,yang2021maximizing} are limited in scope as they apply to only a few possible fairness objectives, and only have weak theoretical guarantees. \citet{do2021two} show how to apply the Frank-Wolfe algorithm to general smooth and concave objective functions for ranking. However, they only solve the problem in a \textit{batch} setting, i.e., computing the recommendations of all users at once, which makes the algorithm impractical for large problems, because of both computation and memory costs. Our algorithm can be seen as an online variant of this algorithm, which resolves all scalability issues.

We showcase the generality of \ofr on three running examples of objective functions for fairness of exposure. The first two objectives are welfare functions for two-sided fairness \citep{do2021two}, and the criterion of quality-weighted exposure \citep{singh2018fairness,biega2018equity}. The third objective, which we call \textit{balanced exposure to user groups}, is novel. Taking inspiration from audits of job advertisement platforms \citep{imana2021auditing}, this objective considers maximizing ranking performance while ensuring that each item is evenly exposed to different user groups defined by sensitive attributes.

In the remainder of the paper, we present the recommendation framework and the different fairness objectives we consider in the next section. In Sec.~\ref{sec:online}, we present our online algorithm in its most general form, as well as its regret bound. In Sec.~\ref{sec:algorithms}, we instantiate the algorithm on three fairness objectives and provide explicit convergence rates in each case. We present our experiments in Sec.~\ref{sec:xp}. We discuss the related work in Sec.~\ref{sec:relatedwork}. Finally, in Sec.~\ref{sec:discussion}, we discuss the limitations of this work and avenues for future research. 
}

{
\section{Fairness of exposure in rankings}
\label{sec:framework}

This paper addresses the online ranking problem, where users arrive one at a time and the recommender system produces a ranking of $k$ items for that user. We focus on an abstract framework where the recommender system has two informal goals. First, the recommended items should be relevant to the user. Second, the exposure of items should be distributed ``fairly'' across users, for some definition of fairness which depends on the application context. We formalize these two goals in this section, by defining objective functions composed of a weighted sum of two terms: the \textit{user objective} which depends on the ranking performance from the user's perspective, and the \textit{fairness objective}, which depends on the exposure of items. In this section, we focus on the \ideal objective functions, which are defined in a \static ranking framework. In the next sections, we focus on the online ranking setting, where at each time step, an incoming user requests recommendations and the recommender systems produces the recommendation list on-the-fly while optimizing these ideal objective functions. 

In order to disentangle the problem of learning user preferences from the problem of generating fair recommendations, we consider that user preferences are given by an oracle. We start this section by describing the  recommendation framework we consider. We then present the fairness objectives we focus on throughout the paper. 

\paragraph{Notation} Integer intervals are denoted within brackets, i.e., $\forall n\in\Nat$, $\intint{n}=\{1, .., n\}$. 
%The $(n-1)$-dimensional probability simplex is denoted by $\simpn$. 
We use the Dirac notation $\dotp{x}{y}$ for the dot product of two vectors of same dimension $x$ and $y$. Finally, $\indic{{\rm expr}}$ is $1$ when  ${\rm expr}$ is true, and $0$ otherwise.

\subsection{Recommendation framework}

We consider a recommendation problem with $n$ users and $m$ items. We identify the set of users with $\intint{n}$ and the set of items with $\intint{m}$. We denote by $\muij \in[0,1]$ the value of recommending item $j$ to user $i$ (e.g., a rating normalized in $[0,1]$). To account for the fact that users are more or less frequent users of the platform, we define the \textit{activity} of user $i$ as a weight $\wi\in[0,1]$. We consider that $\w=(\w_1, ..., \w_n)$ is a probability distribution, so that in the online setting described later in this paper, $\w_i$ is the probability that the current user at a given time step is $i$.

The recommendation for a user is a \textit{top-$\K$ ranking} (or simply \textit{ranking} when the context is clear), i.e., a sorted list of $\K$ unique items, where typically $\K\ll m$. Formally, we represent a ranking by a mapping $\sigma:\intint{\K}\rightarrow\intint{m}$ from ranks to recommended items with the constraint that different ranks correspond to different items. 
%
% A \emph{top-$\K$} ranking $\sigma:\intint{k}\rightarrow\intint{m}$ maps ranks in $\intint{\K}$ to items in $\intint{m}$. Naturally, a single item appears at most one time in a \emph{top-$\K$} ranking. In the remainder, to lighten the text, we drop the prefix "top-$\K$" when the context is clear. 
%
The ranking performance on the user side follows the \textit{position-based} model, similarly to previous work \citep{singh2018fairness,singh2019policy,patro2020fairrec,biega2018equity,do2021two}. Given a set of non-negative, non-increasing \textit{exposure weights} $\wei=(\wei_1, ..., \wei_\K)$, the ranking performance of $\sigma$ for user $i$, denoted by $\ui(\sigma)$, is equal to:
\begin{align}\label{eq:position_based}
    %  \text{\textit{(position-based model of user-side ranking performance)}}&&
     \ui(\sigma) =  \sum_{\k=1}^{\K} \mu_{i,\sigma(\k)}\wei_\k && \text{with $\wei_1 \geq ... \geq\wei_\K\geq 0$}.
\end{align}
We use the shorthand \textit{user utility} to refer to $\ui$.
% \footnote{The term \textit{utility} sometimes refers to subjective mental state of individuals in the utilitarian literature. Our usage of \textit{utility} follows its broader definition in cardinal welfare economics as a ``measurement of the higher-order characteristic that is relevant to the particular distributive justice problem at hand'' \citep[p. 24]{moulin2003fair}. The utility is our proxy to user satisfaction, which in practice is partly based on historical data such as clicks, likes, shares, comments and so on.} 
%
Following previous work on fairness of exposure, we interpret the weights in $\wei$ as being commensurable to the exposure an item receives given its rank. The weights are non-increasing to account for the \textit{position bias}, which means that the user attention to an item decreases with the rank of the item. 
Given a top-$\K$ ranking $\sigma$, the \textit{exposure vector induced by $\sigma$}, denoted by $\A(\sigma)\in\Re^m$ assigns each item to its exposure in $\sigma$:
\begin{align}
    % \Big(\multiliner{7em}{\textit{exposure vector induced by $\sigma$}} \Big ) && 
    \forall j\in\intint{m}, \A_j(\sigma) = \begin{cases}
    b_{\k}&\text{if~}\exists \k\in\intint{\K}, \sigma(\k)=j\\
    0&\text{otherwise}
    \end{cases}.
\end{align}
The user utility is then equal to $
    \ui(\sigma) = \sum\limits_{j=1}^m \muij\A_j(\sigma) = \dotp{\mui}{\A(\sigma)}$.

In practice, the ranking given by a recommender system to a user is not necessarily unique: previous work in \static rankings consider randomization in their rankings \citep{singh2018fairness}, while in our case of online ranking, it is possible that the same user receives different rankings at different time steps. In that case, we are interested in averages of user utilities and item exposures.
To formally define these averages, we use the notation: 
\begin{equation}
\begin{aligned}
    \arms&=\xset[\big]{\A(\sigma): \sigma \text{~is a~top-$\K$ ranking}}\\ \simpA &= \mathrm{convexhull}(\arms) & \pis &= \simpA^{n}.
    \end{aligned}
\end{equation}
$\arms$ is the set of possible item exposures vectors and $\simpA$ is the set of possible average exposure vectors. $\pis$ is an \textit{exposure matrix}, where $\piij$ is the average exposure of item $j$ to user $i$. 
Under the position-based model, a matrix $\pi\in\pis$ characterizes a recommender system since it specifies the average exposure of every item to every user. We use $\pi$ as a convenient mathematical device to study the optimization problems of interests, keeping in mind that out algorithms effectively produce a ranking at each time step.

Recalling that $\w$ represents the user activities, the user utilities and total item exposures under $\pi$ are defined as
\begin{equation}
\begin{aligned}\label{eq:position_based model_exposures}
    \text{(utility of user $i$)} && \ui(\pi) = \dotp{\mui}{\pii}\\ \text{(exposure of item $j$)}&& \vj(\pi) = \sum_{i=1}^n \wi \piij.
    \end{aligned}
\end{equation}
Fairness of exposure refers to objectives in recommender systems where maximizing average user utility is not the sole or main objective of the system. Typically, the exposure of items $\vj$, or variants of them, should also be taken into account in the recommendation. We formulate the goal of a recommender system as optimizing an objective function $\welf(\pi)$ over $\pi\in\pis$, where $\welf$ accounts for both the user utility and the fairness objectives. 
% We give in the next subsection three examples that we use throughout the paper. We chose these three examples because they showcase the generality of the approach, while relying on similar principles so that the presentation is kept simple.

\subsection{Fairness Objectives}\label{sec:fairness_objectives}

We now present our three examples of objective functions $\welf(\pi)$ in order of ``difficulty'' to perform online ranking compared to static ranking. 
%The principles we outline in this paper are far more general (see Sec.~\ref{sec:discussion} for more discussion on this topic). 
%
In all three cases, it is easy to see that the objective functions are \textit{concave} with respect to the recommended exposures $\pi$. The objective functions should be maximized, so the optimal exposures $\piopt$ satisfy
\begin{equation}
    \pi^* \in \argmax_{\pi\in\pis} \welf(\pi).
\end{equation}
Since our algorithm works on any concave function of the average exposures respecting some regularity conditions, we emphasize that the three objective functions below are only a few examples among many. 

%We present the three examples below, and we discuss in the next section how to achieve near-optimal values of these objectives in online ranking.

\paragraph{Two-sided fairness} The first example is from \citet{do2021two} who optimize an additive concave welfare function of user utilities and item exposures. Interpreting item exposure as the utility of the item's producer, this approach is grounded into notions of distributive justice from welfare economics and captures both user- and item-fairness \citep{do2021two}. 
For $\eta>0$, $\beta>0$ and $\alpha_1\in (-\infty, 1), \alpha_2\in(-\infty, 1)$, the objective function is:
\begin{equation}
\begin{aligned}
    \label{eq:welfobj}
    \welf(\pi) &= \sum_{i=1}^n \w_i\psi_{\alpha_1}\!\big(\ui(\pi)\big) + \frac{\beta}{m} \sum_{j=1}^m \psi_{\alpha_2}\!\big(\vj(\pi)\big) 
    \\
    \text{where~} \psi_{\alpha}(x) &= \begin{cases}
    {\rm sign}(\alpha)(\eta+x)^\alpha&\text{~if~} \alpha\neq 0\\
    \log(\eta+x)&\text{~if~}\alpha=0
    % \\
    % -(\eta+x)^\alpha&\text{~if~}\alpha<0
    \end{cases}.
    \end{aligned}
\end{equation}

Where $\eta>0$ avoids infinite derivatives at 0, $\beta>0$ controls the relative weight of user-side and item-side objectives, and $\alpha_1<1$ (resp. $\alpha_2<1$) controls how much we focus on maximizing the utility of the worse-off users (resp. items) \citep{do2021two}.

\paragraph{Quality-weighted exposure} One of the main criteria for fairness of exposure is \emph{quality-weighted} exposure \citep{biega2018equity,wu2021tfrom} (also called merit-based fairness \citep{singh2018fairness,morik2020controlling}). A measure $\quaj$ of the overall quality of an item is taken as reference, and the criterion stipulates that the item exposure is proportional to its quality. $\quaj$ is often defined as the average value $\muij$ over users. Using this definition of $\quaj$, as noted by \citet{do2021two}, it is possible to optimize trade-offs between average user utility and proportional exposure using a penalized objective of the form:
\begin{equation}
\begin{aligned}\label{eq:quaobj}
    &\welf(\pi) = \sum_{i=1}^n \w_i\ui(\pi) - \beta \sqrt{\eta+\frac{1}{m}\sum_{j=1}^m \Big(\quaavg\vj(\pi)-\quaj\norm{\wei}_1\Big)^2} \\
    &\text{~where~} \quaj = \sum_{i=1}^n \wi\muij \text{~~and~~} \quaavg=\frac{1}{m}\sum_{j=1}^m\quaj.
\end{aligned}
\end{equation}
As before, $\beta>0$ controls the trade-off between user utilities and the fairness of exposure penalty and $\eta>0$ avoids infinite derivatives at $0$. This form of the exposure penalty was chosen because it is concave and differentiable, and it is equal to zero when exposure is exactly proportional to quality, i.e., when $\forall j, j', \frac{\vj}{\quaj} = \frac{\vjp}{\quajp}$. We use $\quaavg\vj(\pi)-\quaj\norm{\wei}_1$ rather than 
%the maybe more natural 
$\frac{\vj(\pi)}{\quaj} -\frac{\norm{\wei}_1}{\quaavg}$ because the the former is more stable when qualities are close to $0$ or estimated.

\paragraph{Balanced exposure to user groups} 
%While the two previous objectives have already been proposed in the context of fairness of exposure, 
We also propose to study a new criterion we call \textit{balanced exposure to user groups}, which aims at exposing every item evenly across different user groups. For instance, a designer of a recommendation system might want to ensure a job ad is exposed to the similar proportion of men and women \citep{imana2021auditing}, or to even proportions within each age category. 
Let $\groups = (\group_1, ..., \group_{\card{\groups}})$ be a set of non-empty groups of users. We do not need $\groups$ to contain all users, and groups may be overlapping. Let $\vgj$ be the exposure of item $j$ within the group $\group$, i.e., the amount of exposure $j$ receives in group $\group$ with respect to the total exposure available for this group. That is, for any $\pi \in \pis$, define
\begin{equation}\nonumber
\vgj(\pi):=\sum_{i\in\group} \frac{\wi}{\wsumg} \piij,\text{ with }\wsumg := \sum_{i\in\group} \wi.
\quad
\vavgj = \frac{1}{\card{\groups}}\sum_{\group\in\groups} \vgj(\pi)
\end{equation}
Also, let $\vavgj := (1/\card{\groups})\sum_{\group\in\groups} \vgj(\pi)$ be the average exposure for item $j$, across all the groups.
The objective function we consider takes the following form, where $\beta>0$ and $\eta>0$ play the same roles as before:
\begin{align}\label{eq:balancedobj}
    % \!\!\resizebox{0.93\linewidth}{!}{
    % $\displaystyle
    \welf(\pi) = \sum_{i=1}^n \w_i \ui(\pi) - \frac{\beta}{m} \sum_{j=1}^m \sqrt{\eta + \sum_{\group\in\groups} \Big(\vgj(\pi)-\vavgj(\pi)\Big)^2}.
    % $
    % }
\end{align}
}

{
\section{Fast online ranking}\label{sec:online}

\subsection{Online ranking}
The \emph{online setting} we consider is summarized as follows. At each time step $t\geq 1$:
\begin{enumerate}
\item 
A user $\it\in\intint{n}$ asks for recommendations. We assume $\it$ is drawn at random from the fixed but unknown distribution of user activities with parameters $\w$, i.e., $\it\sim{\rm Categorical}(\w)$.
\item The recommender system picks a ranking $\sigmat$.
\end{enumerate}
Note that as stated before, the main assumptions of this framework are the fact that incoming users are sampled independently at each step from a distribution that remains constant over time. In our setting, we consider that the (user, item) values $\muij$ are known to the system. However, the user activities $\wi$ are unknown.

% Please see Section~\ref{sec:discussion} for more discussion on this topic. 

Let $\at=\A(\sigmat)$ be the exposure vector induced by $\sigmat$, and define, for every user $i$:
\begin{itemize}
    \item The user counts at time $t$:  $\displaystyle\cti = \sum_{\tp \le t}\indic{\itp=i}$;
    \item The average exposure at time $t$: 
    $\displaystyle \pihit=\frac{1}{\cti}\sum_{\tp \le t}\indic{\itp=i} \atp$.
\end{itemize}
% \begin{align}\label{eq:def:pih}
%     % \Big(\substack{\displaystyle\text{\textit{user count }}\\\displaystyle\text{\textit{up to time $t$}}}\Big)
%     \Big(\multiliner{5em}{\textit{user counts \\at time $t$}}\Big)
%     &&  \cti = \sum_{\tp \le t}\indic{\itp=i} && \Big(\multiliner{8em}{\textit{ avg. exposure for user $i$ at time $t$}}\Big) && \pihit=\frac{1}{\cti}\sum_{\tp \le t}\indic{\itp=i} \atp
% \end{align}

Given an objective function $\welf$ such as the ones defined in the previous section, our goal is to design computationally efficient algorithms with low \textit{regret} when $t$ grows to infinity. More formally the goal of the algorithm is to guarantee:
\begin{align}\label{eq:regret}
    \regt = \max_{\pi\in\pis}\big[\welf(\pi)\big] - \expect[\welf(\piht)] \xrightarrow[t\to\infty]{} 0
    %= O\Big(\frac{1}{\sqrt{t}}\Big),
\end{align}
where the expectation in $\regt$ is taken over the random draws of $\ione, ..., \it$ and the $O(.)$ hides constants that depend on the problem, such as the number of users or items.

\subsection{The \ofr algorithm}

We describe in this section our generic algorithm, called \ofr for \texttt{O}nline \texttt{F}rank-Wolfe for \texttt{F}air \texttt{R}anking. 
\ofr works with an abstract objective function 
% $\obj:(\w,\pi)\in\simpn\times\Pi\mapsto \obj(\w,\pi)\in\Re$, where $\simpn$ is the probability simplex over $\intint{n}$ (user indexes). Notice that we make the objective function explicitly depend on the user activities $\w$ to better describe the practical algorithms later. We denote by
$\objw:\pis\rightarrow\Re$, which is parameterized by the vector of user activities $\w$. The $\objw(\pi)$ of this section is exactly the $\welf(\pi)$ of the previous section, except that we make explicit the dependency on $\w$ because it plays an important role in the algorithm and its analysis.

\paragraph{Assumptions on $\w$ and $\objw$} In the remainder, we assume $\w$ is fixed and non-degenerate, i.e., $\forall i\in\intint{n}, \wi>0$. We assume that for every $\w$, $\pi\mapsto\objw(\pi)$ is \textit{concave} and \textit{differentiable}. More importantly, the fundamental object in our algorithm are the partial derivatives of $\obj$ with respect to $\pi_i$ normalized by $\wi$. Given a user index $i$, let 
\begin{align}
    \pdevwi(\pi) = \frac{1}{\wi}\frac{\partial \objw}{\partial  \pi_i}(\pi) \in \Re^m.
\end{align}
We assume that $\pdevwi$ is bounded and Lipschitz with respect to $\pi_i$: for every $\pi\in\simpA^n$ and every $\pipi\in\simpA$, we have:
\begin{itemize}
    \item Bounded gradients: $\norm{\pdevwi(\pi)}_\infty \leq \boundigw$;
    \item Lipschitz gradients: 
\end{itemize}
\begin{equation*}
    \norm[\big]{\pdevwi(\pi) - \pdevwi(\pione, \ldots, \piimo, \pipi, \piipo, \ldots, \pin)}_2 \leq \lipgwi\norm{\pii-\pipi}_2.
\end{equation*}
Notice that with the normalization by $\wi$ in $\pdevwi$, these assumptions guarantee that the importance of a user is commensurable with their activity, i.e., that the objective does not depend disproportionately on users we never see.
\paragraph{Online Frank-Wolfe with an approximate gradient} 
Our algorithm is described by the following rule for choosing $\at$. First, we rely on $\ghit$  which is an approximation of the gradient $\pdevwi(\piht)$. We describe later the required properties of the approximation we need is given (see Th.~\ref{thm:boundgeneral} below), and the approximation one we use in practice (see \eqref{eq:which_approximate_gradient} below). Notice that we rely on an approximation because the user activities are unknown. Then, choose $\at$ as:
\begin{align}\label{eq:fwmain}
    \at\in\argmax_{\a\in\arms} \dotp{\ghitt}{\a}.
\end{align}
Since we compute a maximum dot product with a gradient (or an approximation thereof), our algorithm is a variant of online Frank-Wolfe algorithms. We discuss in more details the relationship with this literature in Sec.~\ref{sec:relatedwork}. 

Frank-Wolfe algorithms shine when the $\argmax$ in \eqref{eq:fwmain} can be computed efficiently. As previously noted by \citet{do2021two} who only study \static ranking, Frank-Wolfe algorithms are particularly suited to ranking because \eqref{eq:fwmain} only requires a top-$\K$ sorting. Let $\topk(x)$ be a routine that returns the indices of $\K$ largest elements in vector $x$.\footnote{Formally, $\sigma=\topk(x) \implies \Big(x_{\sigma(1)} \geq ... \geq x_{\sigma(\K)}$ and $\forall j\not\in\xset[\big]{\sigma(1), ..., \sigma(K)}, x_j \leq x_{\sigma(\K)}\Big)$, using an arbitrary tie breaking rule as it does not play any role in the analysis.} We have:  
\begin{proposition} 
\protect{\citep[Thm. 1]{do2021two}}
\label{prop:sort}
\begin{equation*}\sigmat = \topk\big(\ghit\big) \implies \A\big(\sigmat\big) \in \argmax_{\a\in\arms} \dotp{\ghit}{\a}.
\end{equation*}
\end{proposition}

We call \ofr (Online Frank-Wolfe for Fair Ranking) the usage of the online Frank-Wolfe update \eqref{eq:fwmain} in ranking tasks, i.e., using Prop. \ref{prop:sort} to efficiently perform the $\argmax$ computation of \eqref{eq:fwmain}.

We are now ready to state our main result regarding the convergence of \dynamic ranking. The result does not rely on the specific structure of ranking problems. The result below is valid as long as $\arms\subset \Re^m$ is a finite set with $\forall \a\in\arms,  0\leq \a_j\leq 1$. We denote by $\boundA = \max_{\a\in\arms} \norm{\a}_1$ ($\boundA = \norm{\wei}_1$ in our case). 
\begin{theorem}[Convergence of \eqref{eq:fwmain}]
\label{thm:boundgeneral} 
Let $\pihz\in\simpA^n$, and assume there exists $\boundegi$ such that  $\forall t\geq 1$ and $\forall i\in\intint{n}$, we have:
\begin{align}\label{eq:approximategradient}
    %\text{\emph{(approximate normalized gradient)}} && 
    % \forall i\in\intint{n}, \forall t\geq 1:
    % ~~~
    \expect\Big[
        \norm[\Big]{\ghit - \pdevwi\big(\pihtmo\big)}_\infty
    \Big] \leq \frac{\boundegi}{\sqrt{t}},
\end{align}
where the expectation is taken over $\ione, \ldots, \itmo$. 

Then, with $\at$ chosen by \eqref{eq:fwmain} at all time steps, we have $\forall t\geq 1$:
\begin{align} 
   \regt \leq 2\boundA\sum_{i=1}^n (\lipgwi+ \boundigw)\frac{\ln(e t)}{t} + \frac{6\boundA \sum_{i=1}^n \sqrt{\w_i}(\boundigw+\boundegi)}{\sqrt{t}}
   \label{eq:generalregret}
\end{align}
\end{theorem}
Appendix A is devoted to the proof of this result. The main technical difficulty comes from the fact that we only update the parameters of the incoming user $\it$ with possibly non-uniform user activities, and we need a stochastic step size $1/\cti$ so that the iterates of the optimization algorithm match the resulting average exposures. Notice that the guarantee does not depend on the choice of $\pihz$ because it only affects the first gradient computed for the user. In practice we set $\pihz$ to the average exposure profile of a random top-$\K$ ranking.

Since we do not have access to the exact gradient because user activities are unknown, we use in practice the approximate gradient built using the empirical user activities:
\begin{align}\label{eq:which_approximate_gradient}
    \ghit = \pdevwhtmoi\big(\pihtmo\big) &&  \text{where~~} \wht=\frac{\ct}{t}
\end{align} 
with a fallback formula when $\whitmo=0$. In the next section, we discuss the computationally efficient implementation of this rule for the three objectives of Sec.~\ref{sec:fairness_objectives}, and we provide explicit bounds for $\boundegi$ of \eqref{eq:approximategradient} in each case.
}

{
\section{Applications of \ofr} \label{sec:algorithms}

\begin{table*}[t]
\begin{tabularx}{\textwidth}{*{3}{>{\centering\arraybackslash}X}}
\begin{algorithmic}
\State \textbf{Input at time $t$:} user index $i$
\State \textit{// Step (1)}
\State \textbf{Compute score for each item $j$:}
 \State 
 $\vphantom{\displaystyle\quahj\frac{\norm{\wei}_1}{\quahavg}\Big)}$
$\displaystyle\ghij = \psi'_{\alpha_1}(\uhi)\muij + \frac{\beta}{m}\psi'_{\alpha_2}(\vhj)$
\State \resizebox{0.65\linewidth}{!}{\color{white}$\displaystyle Z = \sqrt{\eta+\frac{1}{m} \sum_{i=1}^m \Big(\vhj-\frac{\quahj\norm{\wei}_1}{\quahavg}\Big)^2}$}
\State \textit{// Step (2)}
\State \textbf{Update} $\uhi$ and all $\vhj$ using \eqref{eq:welfupdate}
\State \textit{// Step (3)}
\State \textbf{Return} $\sigma = 
% \textrm{top-$\K$}
\topk(\ghi)$
\end{algorithmic}
\captionof{algorithm}{\ofr/two-sided fairness \eqref{eq:welfobj}.\label{alg:twosided}}
&
\begin{algorithmic}
\State \textbf{Input at time $t$:} user index $i$
\State \textit{// Step 1}
\State \textbf{Compute score for each item $j$:}
\State
$\displaystyle\ghij = \muij - \frac{\beta}{m\hat{Z}}\Big(\quahavg\vhj - \quahj\norm{\wei}_1\Big)$
% \State \textbf{Compute normalization}:
% \State $\displaystyle Z = \sqrt{\eta+\frac{1}{m} \sum_{i=1}^m \Big(\vhj-\quahj\frac{\norm{\wei}_1}{\quahavg}\Big)^2}$
\State where \resizebox{0.75\linewidth}{!}{$\displaystyle \hat{Z} = \sqrt{\eta+\frac{1}{m} \sum_{j=1}^m \Big(\quahavg\vhj-\quahj\norm{\wei}_1\Big)^2}$}
\State \textit{// Step (2)}
\State \textbf{Update} all $\vhj, \quahj, \quahavg$ using \eqref{eq:quaupdate}
\State \textit{// Step (3)}
\State \textbf{Return} $\sigma = 
% \textrm{top-$\K$}
\topk(\ghi)$
\end{algorithmic}
\captionof{algorithm}{\ofr/quality-weighted  \eqref{eq:quaobj}.\label{alg:quaexpo}}
&
\begin{algorithmic}
\State \textbf{Input at time $t$:} user index $i$
\State \textit{// Step 1 ($\groupofi$ is the group of user $i$)}
\State \textbf{Compute score for each item $j$:}
\State \resizebox{\linewidth}{!}{$\displaystyle\ghij = \muij - \frac{\beta}{m \hat{Z}_j}.\frac{t}{\cgi+1}\big(\vhgij - \vhavgj\big)$}
% \State \textbf{Compute normalizations}:
% \State $\displaystyle Z_j = \sqrt{\eta + \sum_{\group\in\groups} \Big(\vhgj-\vhavgg\Big)^2}$
\State where  \resizebox{0.8\linewidth}{!}{$\displaystyle \hat{Z}_j = \sqrt{\eta + \sum_{\group\in\groups} \Big(\vhgj-\vhavgj\Big)^2}$}
\State \textit{// Step (2)}
\State \textbf{Update}  $\forall j, \vhgij$ using \eqref{eq:balancedupdate}
\State \textit{// Step (3)}
\State \textbf{Return} $\sigma = 
% \textrm{top-$\K$}
\topk(\ghi)$
\end{algorithmic}
\captionof{algorithm}{\hspace{-0.5mm}\ofr/balanced exposure \eqref{eq:balancedobj}.\label{alg:balancedexpo}}
\end{tabularx}
\end{table*}

Practical implementations of \ofr do not rely on naive computations of $\ghit = \pdevwhtmoi\big(\pihtmo\big)$, because they would require explicitly keeping track of $\piht$. $\piht$ is a matrix of size $n\times m$, which is impossible to store explicitly in large-scale applications.\footnote{$n\times m$ is also the size of the matrix of (user, item) values $\mu$, which in practice is not stored explicitly. Rather, the values $\muij$ are computed on-the-fly (possibly using caching for often-accessed values) and the storage uses compact representations, such as latent factor models \citep{koren2015advances} or neural networks \citep{he2017neural}.} Importantly, as we illustrate in this section, for the objectives of Sec.~\ref{sec:fairness_objectives} it is unnecessary to maintain explicit representations of $\piht$ because the gradients depend on $\piht$ only through utilities or exposures, for which we can maintain online estimates.

\subsection{Practical implementations}
The implementation of \ofr for the three fairness objectives \eqref{eq:welfobj}, \eqref{eq:quaobj} and \eqref{eq:balancedobj} are described in Alg.~\ref{alg:twosided}, \ref{alg:quaexpo} and \ref{alg:balancedexpo} respectively, where we dropped the superscripts ${}^{(t-1)}$ and ${}^{(t)}$ for better readability. 
At every round $t$, there are three steps:
\begin{enumerate}
    \item compute approximate gradients based on online estimates of user values and exposures,
    \item update the relevant online estimates of user utility and item exposures,
    \item perform a top-$\K$ sort of the scores computed in step (1) to obtain $\sigmat$.
\end{enumerate}

We omit the details of the calculation of $\ghij$ in Alg.~\ref{alg:twosided}, \ref{alg:quaexpo} and \ref{alg:balancedexpo}, which are obtained by differentiation of $\objwh$ using \eqref{eq:which_approximate_gradient}.\footnote{In Alg.~\ref{alg:balancedexpo}, we use a factor $\frac{t}{\cgi+1}$ while the direct calculation would give a factor $\frac{t}{\cgi}$. The formula we use more gracefully deals with the case $\cgi=0$ and enjoys similar bounds when $t$ is large.}

\paragraph{Two-sided Fairness} For two-sided fairness \eqref{eq:welfobj}, we have:
\begin{align}\pdevwij\big(\piht\big) = \psi'_{\alpha_1}\big(\ui(\piht\big)\mui + \frac{\beta}{m}\psi'_{\alpha_2}\big(\vj(\piht)\big).
\end{align}
Let:
\begin{equation}
\begin{aligned}%\label{eq:welfupdate}
    \forall i\in\intint{n}, \uhit &= \ui\big(\piht\big) = 
    % \dotp{\mui}{\piht} =  
    \frac{1}{\cti}\sum_{\tp \le t} \indic{\itp=i} \dotp{\mui}{\atp},
    \\
    \forall j\in\intint{m},  \vhjt &= \sum_{i=1}^n \whit \pihijt = \frac{1}{t}\sum_{\tp=1}^t \atpj.
\end{aligned}
\end{equation}
(Recall $\at=\A(\sigmat)$.) This gives the formula computed in Alg.~\ref{alg:twosided} for $\ghit = \pdevwhtmoi\big(\pihtmo\big)$.

%Alg.~\ref{alg:twosided} for two-sided fairness is then obtained. 
For the online updates of $\uh$ and $\vh$, we use as initial value $\uhiz$ the utility of the random ranking
% \footnote{The formula is $\forall i, \uhiz=\dotp{\mui}{\frac{\norm{\wei}_1}{m}}$} 
and $\vhjz=0$. Since $\uhit$ only changes for $\it=i$, they are given by:
\begin{equation}
\begin{aligned}\label{eq:welfupdate}
   \forall i, \uhiz&=\dotp{\mui}{\frac{\norm{\wei}_1}{m}} \\
   \uhitt  &=  \uhittmo +\frac{1}{t}\Big(\dotp{\mui}{\at} - \uhittmo \Big) \\
   \vht &=  \vhtmo+\frac{1}{t}\Big(\at - \vhtmo\Big).
\end{aligned}
\end{equation}
% \begin{align}
%   \forall i, \uhiz=\dotp{\mui}{\frac{\norm{\wei}_1}{m}} && \uhitt  =  \frac{t-1}{t}\uhittmo +\frac{1}{t}\dotp{\mui}{\at} && \vht =  \frac{t-1}{t}\vhtmo+\frac{1}{t}\at
% \end{align}

\paragraph{Quality-weighted exposure} Similarly, for quality-weighted exposure \eqref{eq:quaobj}, approximate gradients $\pdevwhti(\piht)$ use online estimates of exposures $\vht$ as in \eqref{eq:welfupdate}, as well as online estimates of the qualities using $\forall j, \quahjz = 0$:
\begin{equation}
\begin{aligned}\label{eq:quaupdate}
    \quaht &= \frac{1}{t} \sum_{\tp=1}^t \muit = \quahtmo +  \frac{1}{t}\Big(\muit - \quahtmo \Big), \\
    \quahavgt &= \frac{1}{m} \sum_{j=1}^m \quahjt.
\end{aligned}
\end{equation}

\paragraph{Balanced exposure} Balanced exposure to user groups \eqref{eq:balancedobj} works similarly, except that we need to keep track of user counts within each group, which we denote by $\ctg$, as well as exposures within each group:
\begin{equation}
\begin{aligned}\label{eq:balancedupdate}
\forall j\in\intint{m},\ctg &= \sum_{i\in\group} \cti,\\ 
\vhgjt &= \frac{1}{\ctg}\sum_{\tp=1}^t \indic{\itp\in\group} \atpj, \\ \vhavgjt &= \frac{1}{\card{\groups}} \sum_{\group\in\groups} \vhgjt.
\end{aligned}
\end{equation}
We use $\vhgjt = 0$ if $\ctg=0$ since the item has not been exposed to a group we never saw. 
As for $\vh$ and $\quah$, these counts are updated online in $O(m)$ operations because they only change for the group of user $\it$.

The guarantees we obtain for these algorithms are the following. The proof is given in App.~B.
\begin{proposition}\label{prop:approx_gradients} The approximate gradients of Alg.~\ref{alg:twosided}, \ref{alg:quaexpo} and \ref{alg:balancedexpo} satisfy:
% \begin{itemize}
%     \item[Alg.~\ref{alg:twosided}:] $\displaystyle \norm{\pdevwhti(\piht) - \pdevwi\big(\piht\big)}_\infty \leq  \frac{\beta \norm{\psi''_{\alpha_2}}_{\infty}}{m}\sqrt{\frac{n}{t}}$,
%     \item[Alg.~\ref{alg:quaexpo}] $\displaystyle \norm{\pdevwhti(\piht) - \pdevwi\big(\piht\big)}_\infty \leq \frac{\beta\big(2+\norm{\wei}_1\big)^2}{m\min(\eta,\sqrt{\eta})}\sqrt{\frac{n}{t}}$,
%     \item[Alg.~\ref{alg:balancedexpo}] $\displaystyle \norm{\pdevwhti(\piht) - \pdevwi\big(\piht\big)}_\infty \leq \frac{\beta}{m\wsumgofi\min(\eta,\sqrt{\eta})}\bigg(
%     \frac{1}{\wsumgofi(t+1)}
%     +
%     8\sum_{\group\in\groups}\sqrt{\frac{\card{\group}}{\wsumg t}}\bigg)$.
% \end{itemize}
\begin{itemize}[leftmargin=1.5cm]
    \item[Alg.~\ref{alg:twosided}:] $\displaystyle \expect\Big[
        \norm[\Big]{\ghit - \pdevwi\big(\pihtmo\big)}_\infty
    \Big] \leq  \frac{\beta \norm{\psi''_{\alpha_2}}_{\infty}}{m}\sqrt{\frac{n}{t-1}}$,
    \item[Alg.~\ref{alg:quaexpo}:] $\displaystyle \expect\Big[
        \norm[\Big]{\ghit - \pdevwi\big(\pihtmo\big)}_\infty
    \Big] \leq \frac{\beta\big(2+\norm{\wei}_1\big)^2}{m\min(\eta,\sqrt{\eta})}\sqrt{\frac{n}{t-1}}$,
     \item[Alg.~\ref{alg:balancedexpo}:] $\displaystyle \expect\Big[
        \norm[\Big]{\ghit - \pdevwi\big(\pihtmo\big)}_\infty
    \Big] \leq \frac{\beta\Big(
    \frac{1}{\wsumgofi t}
    +
    8\sum_{\group\in\groups}\sqrt{\frac{\card{\group}}{\wsumg (t-1)}}\Big)}{m\wsumgofi\min(\eta,\sqrt{\eta})}$.
\end{itemize}
\end{proposition}
Overall, they all decrease in $O(\frac{1}{\sqrt{t}})$ as desired to apply our convergence result Th.~\ref{thm:boundgeneral}. More interestingly, the bounds do not depend on $\w$, which means that the objectives are well-behaved even when some users have low probabilities. The balanced exposure criterion does depend on $\frac{1}{\wsumg}$, which means that the bound becomes arbitrarily bad when some groups have small cumulative activity. This is natural, since achieving balanced exposure across groups dynamically is necessarily a difficult task if one group is only very rarely observed.

% Let $\overline{w}_{\min}$ be the minimum total weight of a group, that is
% \begin{equation}
% \overline{w}_{\min} := \min_{s \in \mathcal S}\wsumg.
% \label{eq:wmin}
% \end{equation}

Putting together Thm.~\ref{thm:boundgeneral} and Prop.~\ref{prop:approx_gradients}, we obtain regret bounds of order $1/\sqrt{t}$:
\begin{corollary}%[Convergence rate of proposed algorithms]
\label{cor:boundalgos}
Ignoring constants and assuming $\eta\leq 1$, the regrets $R(t)$ of Alg.~\ref{alg:twosided}, ~\ref{alg:quaexpo} and \ref{alg:balancedexpo} are bounded as in Table~\ref{tab:rates}.
\end{corollary}

\begin{table}[h!]
  \begin{center}
       {\tabulinesep=1.2mm
    \begin{tabu}{|c|c|c|}
    \hline
    %   Problem &  
      Algorithm & Order of magnitude of $\regt$\\
      \hline
    %   Two-sided fairness & 
      Alg.~\ref{alg:twosided} & $\bigg(\sqrt{n}\|b\|_1 + \dfrac{n\|b\|_1\beta}{m}\bigg)\Big(\eta^{\alpha_1-1} + \eta^{\alpha_2-2}\Big)\sqrt{\dfrac{1}{t}}$\\
      \hline
    %   Quality-weighted & 
      Alg.~\ref{alg:quaexpo} & $\bigg(\sqrt{n}\|b\|_1 + \dfrac{n\|b\|_1^3\beta}{m}\bigg)\eta^{-1}\sqrt{\dfrac{1}{t}}$\\
      \hline
    %   Balanced exposure & 
      Alg.~\ref{alg:balancedexpo} 
      & $\bigg(\sqrt{n}\|b\|_1 + \dfrac{n\|b\|_1\beta}{m}\sqrt{\dfrac{|\mathcal S|}{\overline{w}_{\min}^3}}\bigg)\eta^{-1}\sqrt{\dfrac{1}{t}}$\\
      \hline
      \end{tabu}
      }
      \end{center}
      \caption{Upper bounds on regret $R(t)$, ignoring constants and assuming $\eta\leq 1$. In all cases, we have $R(t) = \mathcal O(1/\sqrt{t})$. For balanced exposure, the regret bound also depends on the minimum total weight of a group $\overline{w}_{\min}= \min_{s \in \mathcal S}\wsumg$. }
      \label{tab:rates}
\end{table}

The proof of Cor.~\ref{cor:boundalgos} is given in App.~C. Compared to a batch Frank-Wolfe algorithm for the same objectives, we obtain a convergence in $O(1/\sqrt{t})$ instead of $1/t$ \citep[Prop 4.]{do2021two}. Part of this difference is due to the variance in the gradients due to unknown user activities, but Th.~\ref{thm:boundgeneral} would be of order $1/\sqrt{t}$ even with true gradients (i.e., $\boundegi=0$). We do not believe our bound can be improved because our online setting we consider is only equivalent to a Frank-Wolfe algorithm if we consider a stochastic ``stepsize'' of $1/\cti$ for user $i$ at time step $t$ (which yields the average exposures $\piht$, our object of study), which introduces additional variance in the optimization. We leave the proof of lower bounds to future work.

\subsection{Computational complexity} 

To simplify the discussion on computational complexity, we assume the number of groups in balanced exposure is $O(1)$ (i.e., negligible compared to $n$ and $m$), which is the case in practice for groups such as gender or age. For each of the algorithms, the necessary computations involve $O(m)$ floating point operations to compute the scores of all items (step (1)), $O(m+\K\ln\K)$ (amortized) comparisons for the top-$\K$ sort (step (3)). The update of the online estimates (step (2)) requires $O(m)$ operations. 
% \footnote{
More involved implementations of this step require only $O(\K)$ operations by only updating the recommended items, but they require additional computations in step (1), which remains in $O(m)$.
% } 
In all cases, the computation cost is dominated by the top-$\K$ sort, which would likely required in practice even without consideration for fairness of exposure. Thus, \ofr provides a general approach to fairness of exposure in online ranking that does not involve significantly more computations than having no considerations for fairness of exposure at all, despite optimizing an objective function where the optimal ranking of each user depends on the rankings of all other users.

\paragraph{Memory requirements} For two-sided fairness (Alg.~\ref{alg:twosided}), we need $O(n+m)$ bytes for $\uh$ and $\vh$. For quality weighted exposure we need  $O(m)$ bytes to store $\vh$ and $\quah$, while we need $O(m\card{\groups})$ bytes for balanced exposure. In all cases, storage is of the order $O(n+m)$. Notice that in practice, it is likely that counters of item exposures and user utility are computed to monitor the performance of the system anyway. The additional storage of our algorithm is then negligible.
}
{
\section{Experiments}\label{sec:xp}

We provide in this section experiments on simulated ranking tasks following the protocol of \citet{do2021two}. Our experiments have two goals. First we study the convergence of \ofr to the desired trade-off values for the three objectives of Sec.~\ref{sec:fairness_objectives}, by comparing objective function values of \ofr and the batch Frank-Wolfe algorithm for fair ranking of \citep{do2021two} at comparable computational budgets. Second, we compare the dynamics of \ofr and \fairco \citep{morik2020controlling}, an online ranking algorithm designed to asymptotically achieve equal quality-weighted exposure for all items. We also provide a comparison between \ofr and \fairco on balanced exposure by proposing an ad-hoc extension to \fairco for that task. In the three next subsections, we first describe our experimental protocol (Subsection \ref{sec:xp:setup}). Then, we give qualitative results in terms of the trade-offs achieved by \ofr by varying the weight of the exposure objective (Subsection~\ref{sec:xp:quali}). We finally we dive into the comparison between \ofr and batch Frank-Wolfe (Subsection~\ref{sec:xp:results_convergence}) and between \ofr and \fairco (Subsection~\ref{sec:xp:results_fairco}).

\subsection{Experimental setup}\label{sec:xp:setup}

\paragraph{Data} We  use the Last.fm dataset of \citet{Celma:Springer2010}, which includes $360k$ users and $180k$ items (artists), from which we select the top $15k$ users and $15k$ items having the most interactions. We refer to this subset of the dataset as \lastfm. The (user, item) values are estimated using a standard matrix factorization for learning from positive feedback only \citet{hu2008collaborative}. Details of this training can be found in App.~D.1. Since we focus on ranking given the preferences rather than on the properties of the matrix factorization algorithm, we consider these preferences as our ground truth and given to the algorithm, following previous work\citep{patro2020fairrec,wu2021tfrom,chakraborty2019equality,do2021two}.  In App.~D.3, we present results on the MovieLens dataset \cite{harper2015movielens}. The results are qualitatively similar. Both datasets come with a ``gender'' table associated to user IDs. It is a ternary value 'male', 'female', 'other' (see \citep{Celma:Springer2010,harper2015movielens}  for details on the datasets). On \lastfm, the resulting dataset contains $~10k$/$~3.6k$/$1.4k$ users of category 'male'/'female'/'other' respectively.

\paragraph{Tasks} We study the three tasks described in Sec.~\ref{sec:fairness_objectives}: two-sided fairness, quality-weighted exposure and balanced exposure to user groups. Note that it is possible to study weighted combinations of these objective, since the combined objective would remain concave and smooth. We focus on the three canonical examples to keep the exposition simple. We use the gender category described above as user groups, and they are only used for the last objective. We study the behavior of the algorithm as we vary $\beta>0$, which controls the trade-off between user utility and item fairness. For two-sided fairness, we take $\alpha_1 = \alpha_2=0$ in \eqref{eq:welfobj}, which are recommended values in \citet{do2021two} to generate trade-offs between user and item fairness. In all cases, we use  $\eta=1$ in this section, and show the results for $\eta=0.01$ (a less smooth objective function) in App.~D.2. We assume that the true user activities are uniform (but the online algorithm does not know about these activities). We consider top-$\K$ rankings with $\K=40$. We set the exposure weights to $\wei_\k=\frac{1}{\log_2(1+\k)}$, which correspond to the well-known DCG measure, as in \citep{biega2018equity,patro2020fairrec,do2021two}. In the following, we use the term \textit{iteration} to refer to a time step, and \textit{epoch} to refer to $n$ timesteps (which correspond to the order of magnitude of time steps required to see every user). Notice that in the online setting, users are sampled with replacement at each iteration following our formal framework of Sec.~\ref{sec:framework}, so the online algorithms are not guaranteed to see every user at every epoch.

\paragraph{Comparisons} We compare to two previous works:
\begin{enumerate}[leftmargin=*]
    \item Batch Frank-Wolfe (\batchfw): The only tractable algorithm we know of for all these objectives is the \static algorithm of \citet{do2021two}. We compare offline vs online learning in terms of convergence to the objective value for a given computational budget. The algorithm of \citet{do2021two} is based on Frank-Wolfe as well, which we refer to as \batchfw and our approach (referred to as \ofr) is an online version of \batchfw. Thus the cost per user  per epoch (one top-$\K$ sort) are the same. We use this baseline to benchmark how fast \ofr convergence to the optimal value.
    \item \fairco \citep{morik2020controlling}: We use the approach from \citep{morik2020controlling} introduced for quality-weighted exposure in dynamic ranking. In our notation, dropping the time superscrits, given user $i$ at time step $t$, \fairco outputs
    \begin{equation}
    \begin{aligned}\label{eq:fairco}
        \text{(\fairco \citep{morik2020controlling})}&&
        \sigma &= \topk(\tilde{\mu}) \\
        &&\text{with~} \tilde{\mu}_j &= \muij + \beta (t-1) \max_{\jp}\Big(\frac{\vhjp}{\quahjp} - \frac{\vhj}{\quahj} \Big)
    \end{aligned}
    \end{equation}
    where $\beta$ trades-off the importance of the user values $\muij$ and the discrepancy between items in terms of quality-weighted exposure. Notice that the item realizing the maximum in \eqref{eq:fairco} is the same for all $j$, so the computational complexity of \fairco is similar to that of \ofr.
    
    A fundamental difference between \fairco is that the weight given to the fairness objective increases with $t$. The authors proved in the paper that the average $\frac{1}{m(m-1)} \sum_{j, \jp} \Big|\frac{\vhjp}{\quahjp} - \frac{\vhj}{\quahj} \Big|$  converges to $0$ at a rate $O(1/t)$. However, they do not discuss the convergence of the user utilities depending on $\beta$. Fundamentally, \fairco and \ofr address different problems, since \ofr aims for trade-offs where the relative weight of the user objective and the item objective is fixed from the start. Even though they should converge to different outcomes, we compare the intermediate dynamics at the early stage of optimization.

\end{enumerate}
In addition, we compare to an extension of \fairco to balanced exposure. Even though \fairco was not designed for balanced exposure, we propose to follow a similar recipe as \eqref{eq:fairco} as baseline for balanced exposure:
    \begin{align}\label{eq:fairco_balanced}
        \sigma = \topk(\tilde{\mu}) &&\text{with~} \tilde{\mu}_j = \muij + \beta (t-1) \max_{\group\in\groups}\Big(\vhgj - \vhgij \Big)
    \end{align}

All our experiments are repeated and averaged on three seeds for sampling the users at each step. The online algorithms are run for $5000$ epochs, and the batch algorithms for $50,000$ epochs.

% \begin{remark}[\dynamic vs \static ranking] Since \batchfw is applicable to all our objectives, we could think of a baseline for dynamic ranking as:
% \begin{enumerate}
%     \item Optimize the \static objective offline with \batchfw to get optimal $\pih$ using $T$ batch Frank-Wolfe epochs,
%     \item at each time step $t$, sample a ranking from $\pihi$, using Birkhoff-von Neumann decomposition \citep[see e.g.,]{singh2018fairness}.
% \end{enumerate}
% Unfortunately, this procedure requires storing $\pih$, which costs either $O(n\times m\times \K)$ or $O(n\times T\times \K)$ bytes \citep{do2021two}. This is prohibitive when we need $T$ of the order of several hundreds or thousands (as is the case in our experiments. Moreover, applying this baseline requires the knowledge of user activities. In contrast, the storage requirement of \ofr is $O(n+m)$, and it does not require the knowledge of the user activities beforehand.
% \end{remark}

\subsection{Qualitative results: effect of varying $\protect\beta$}\label{sec:xp:quali}

\begin{figure*}[t]
    \centering
    \includegraphics[width=0.85\linewidth]{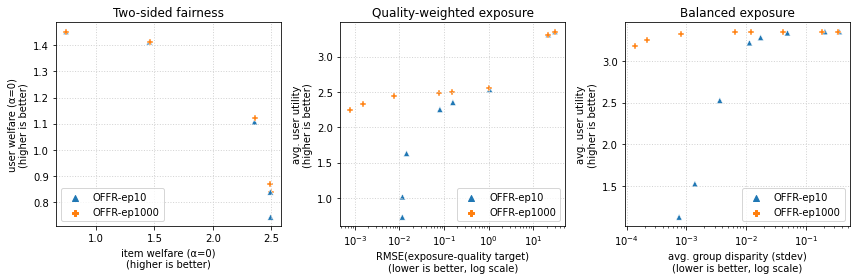}
    \caption{Trade-offs between user objective ($y$-axis) and item fairness ($x$-axis), at the beginning of the online process ($10$ epochs) and closer to the end ($1000$ epochs). For two-sided fairness, both user and item objectives should be maximized, while for quality-weighted balanced exposure the item objectives ($x$-axis) should be minimized. As expected, varying the weight of the item objective $\beta$ leads to different trade-offs between user utility and item exposure. Comparing epochs $10$ and $1000$ on quality-weighted and balanced exposure, we observe that with large $\beta$, \ofr tends to prioritize the item objective and has low user utility at the beginning of training.}
    \label{fig:tradeoffs}
\end{figure*}

We first present qualitative results regarding the trade-offs that are obtained by varying the weight of the fairness penalty $\beta$ from 0.001 to 100 by powers of $10$, for all three tasks in Fig.~\ref{fig:tradeoffs}. The $y$-axis is the user objective for two-sided fairness and the average user utility for quality-weighted and balanced exposure. The $x$-axis is the item objective (higher is better) for two-sided fairness, and the item penalty term with $\eta=0$ of \eqref{eq:quaobj} and \eqref{eq:balancedobj} for quality-weighted and balanced exposure respectively. 

At a high level, we observe as expected a Pareto front spanning a large range of (user, item) objective values. We also observe on all three tasks but specifically quality-weighted and balanced exposure that at the beginning of training (epoch $10$), the item objective values are close to the final values, but the user objective values increase a lot on the course of training. We will get back to this observation in our comparison with \fairco. As anecdotal remarks, we first observe that on two-sided ranking, convergence is very fast and the trade-offs obtained at epoch $10$ and $1000$ are relatively close. Second, we observe that for balanced exposure on this dataset, it is possible to achieve near perfect fairness (item objective $\leq 10^{-3}$) at very little cost of user utility at the end of training.

\subsection{Online convergence}\label{sec:xp:results_convergence}

To compare the convergence of \ofr compared to \batchfw, 
Fig.~\ref{fig:convergence} plots the regret, i.e., the difference between the maximum value and the obtained objective value on the course of the optimization for both algorithms,\footnote{The maximum value for the regret is taken as the maximum between the result of \batchfw after $50$k epochs and \ofr after $5$k epochs. For both \ofr and \batchfw, we compute the ideal objective function, i.e., knowing the user activities $\w$. Notice that \ofr does not know $\w$ but \batchfw does.} in log-scale as a function of the number of epochs for the three tasks for $\beta\in\{0.01, 1.0\}$. We first observe that convergence is slower for larger values of $\beta$, which is coherent with the theoretical analysis. We also observe that for the first $1000$ epochs (recall that an epoch has the same compatational cost for both algorithms), \ofr  fares better than the batch algorithm. Looking at more epochs or different values of $\eta$ (shown in Fig.~4 and 6 in App.~D.2), we observe that \batchfw eventually catches up. This is coherent with the theoretical analysis, as the \batchfw converges in $1/t$ \citep{do2021two}, but \ofr in $O(1/\sqrt{t})$. In accordance with the well-known performance of stochastic gradient descent in machine learning \citep{bottou2007tradeoffs}, the online algorithm seems to perform much better at the beginning, which suggests that it is practical to run the algorithm online.

\begin{figure*}[t]
    \centering
    \includegraphics[width=0.85\linewidth]{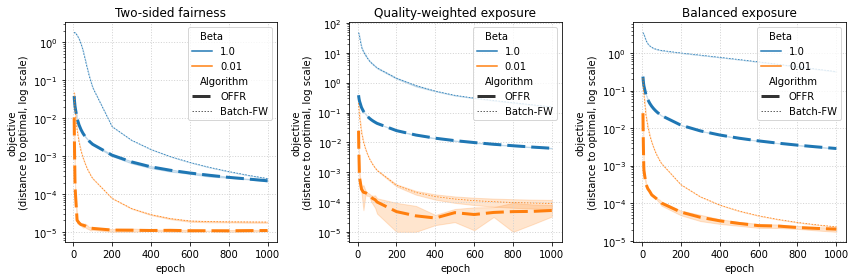}
    \caption{Convergence speed of \ofr compared to \batchfw on the three fairness objectives, for $\beta \in\{0.01, 1\}$ and $\eta=1$. The $y$-axis is the regret in log-scale. We observe that \ofr is faster than \batchfw at the beginning, especially for large $\beta$.}
    \label{fig:convergence}
\end{figure*}

\subsection{Comparison to \fairco}\label{sec:xp:results_fairco}

We give in Fig.~\ref{fig:fairco} illustrations of the dynamics of \ofr compared to \fairco \citep{morik2020controlling} described in \eqref{eq:fairco} and \eqref{eq:fairco_balanced}. Since \fairco aims at driving a disparity to $0$, it cannot be applied to two-sided fairness, so we focus on quality-weighted and balanced exposure. Contrarily to the previous section, we cannot compare objective functions at convergence or convergence rates because \fairco does not optimize an objective function. Our plots show the item objective ($x$-axis, lower is better, log-scale) and the average user utility on the $y$-axis. Then, for \ofr and \fairco for two values of $\beta$, we show the (item objective, user utility) values obtained on the course of the algorithm. For \fairco, we chose $\beta=0.001$ which gives overall the highest user utility values we observed, and $\beta=1$, a representative large value. For \ofr we chose two different values of $\beta$ that achieve different trade-offs. This choice has no impact on the discussion. The left plots show vanilla \ofr, the right plots show \ofr with a ``pacing'' heuristic described later.

\begin{figure*}[t]
    \centering
    \includegraphics[width=0.48\linewidth]{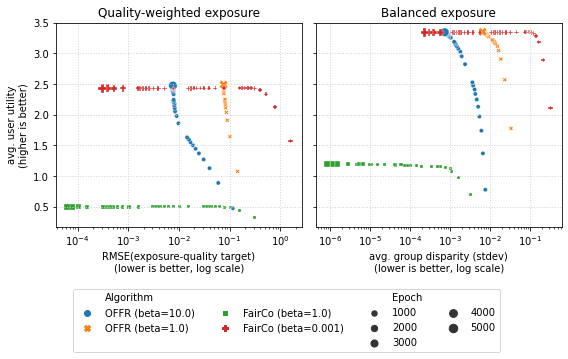}
    \hfill
    \includegraphics[width=0.48\linewidth]{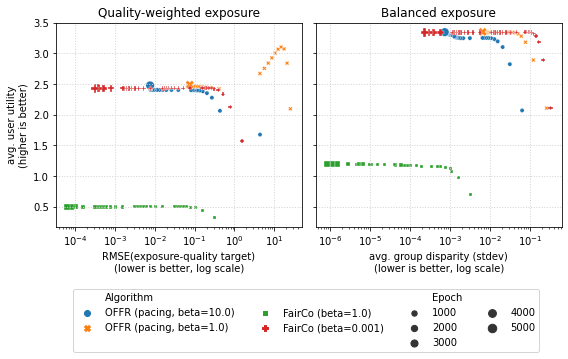}
    \caption{Convergence of \ofr compared to the dynamics of \fairco. Each point is the average user utility ($y$-axis) vs item objective ($x$-axis, log-scale) for an algorithm and value of $\beta$ (color/marker), at a given epoch (the size of the markers increase with the epoch number). The trajectory describes the online dynamics of each algorithm in terms of the trade-offs they achieve. (left) \ofr converges to the trade-off dictated by its value of $\beta$ while keeping its item objective near the target value from the beginning, increasing the user utility with time. (right) the pacing heuristic added to \ofr provides a way to approach the final trade-off while keeping user utility high during the entire course of optimization.}
    \label{fig:fairco}
\end{figure*}

\paragraph{Convergence properties} Looking at the left plot, we see \ofr converging to its trade-off dictated by the value of $\beta$. On the other hand, \fairco does not converge. As expected, as time goes by, \fairco reduces the item objective to low values. Interestingly though, it seems that the average user utility \textit{seems to} converge for \fairco to a value that depends on $\beta$. It is likely an effect of the experimental setup: with $\beta=0.001$, \fairco is far from the regime where it achieves low values of the item objective within our $5000$ epochs (as seen by the discrepancy in item objective between $\beta=1$ and $\beta=0.001$). Overall, since \fairco does not have a clear objective function nor theoretical guarantees regarding the user utility, \fairco does not allow to choose the trade-off between user and item objectives that is desired. On the bright side, \fairco does happen to reduce the item objective to very low values for $\beta=1$ as the number of iteration decreases. 

\paragraph{Trade-offs} Interestingly, \fairco and \ofr have different dynamics. The plots show that on the course of the iterations, \ofr rapidly reaches its item objective, but takes time to reach its user objective (as seen by the ``vertical'' pattern of \ofr in the left plot of Fig.~\ref{fig:fairco}, which means that the item objective does not change a lot). In contrast, \fairco for small $\beta$ starts from high user utility and decreases the item objective from there. Evidently, \ofr and \fairco strike different trade-offs, and neither of them is universally better: it depends on whether we prioritize user utility or item fairness at the early stages of the algorithm. Nonetheless, to emulate \fairco's trade-offs, we propose a ``pacing'' heuristic which uses a time-dependent $\beta$ in our objective, using $\beta_t = \min(\beta, \gamma.\frac{t}{n})$ where $\gamma>0$ is the pacing factor. the right plots of Fig.~\ref{fig:fairco} show the results with $\gamma=0.01$. We observe now a more ``horizontal'' pattern in the dynamics of \ofr, similarly to \fairco, meaning that \ofr sucessfully pioitizes user utility over item fairness in the early stages of the algorithm. Whether or not such a pacing should be used depends on the requirements of the application.

}
{
\section{Related Work}\label{sec:relatedwork}

The question of the social impact of recommender systems started with independent audits of bias against groups defined by sensitive attributes 
\citep{sweeney2013discrimination,kay2015unequal,hannak2014measuring,mehrotra2017auditing,lambrecht2019algorithmic}.
Algorithms for fairness of exposure have been studied since then
\citep{celis2017ranking,burke2017multisided,biega2018equity,singh2018fairness,morik2020controlling,zehlike2020reducing,do2021two}. The goal is often to prevent winner-take-all effects or popularity bias \citep{singh2018fairness,abdollahpouri2019unfairness} or promote smaller producers to incentivize production \citep{liu2019personalized,mehrotra2018towards,mladenov2020optimizing}.

The question of online ranking is often studied in conjunction with learning to rank, i.e., learning the (user, item) values $\mu$. The previous work by \citet{morik2020controlling}, which we compare to in the experiments (the \fairco baseline), had this dimension, which we do not. On the other hand, as we discussed, their algorithm has limited scope because it only aims at asymptotically removing any disparity. The algorithm cannot be used on other forms of loss function such as two-sided fairness, and cannot be used to converge to intermediate trade-offs. Their theoretical guarantee is also relatively weak, since they only prove that the exposure objective converges to $0$, without any guarantee on the user utility. In contrast, we show that the regret of our algorithm converges to 0 for a wide range of objectives.

\citet{yang2021maximizing} also proposes an online algorithm combining learning to rank and fairness of exposure, but they compare the exposure of groups of items within single rankings, as opposed to considering exposure of items across users. Their fairness criterion does not involve the challenge we address, since the optimal rankings in their case can still be computed individually for every user.
%
% Fairness of exposure in dynamic ranking is studied in  \citep{morik2020controlling,yang2021maximizing}, with criteria similar to the quality-weighted exposure penalty. The goal of \citet{morik2020controlling} is to drive disparities in quality-weighted exposure to zero, rather than maximizing a user utility-item unfairness trade-off. These works address the problem of dynamic learning-to-rank, meaning that the position bias and the relevance of items to users must be estimated from user feedback. However, none of these works consider the fact that users have different activity levels, which must be estimated to properly measure the considered fairness objectives. This is the key challenge we address in this paper. %In the same setting, \cite{yang2021maximizing} consider a similar criterion, but focusing on the top-$k$ positions only. 
%

Recently, fairness of exposure has been studied in the bandit setting  \citep{jeunen2021top,mansoury2021unbiased}. These works provide experimental evaluations of bandit algorithms with fairness constraints,  but they do not provide theoretical guarantees. \citet{wang2021fairness} also consider fairness of exposure in bandits, but without ranking.

Compared to this literature on dynamic ranking, we decided to disentangle the problem of learning the user preferences from the problem of generating the rankings online while optimizing a global exposure objective. We obtain a solution to the ranking problem that is more general than what was proposed before, with stronger theoretical guarantees. Our approach unlocks the problem of online ranking with a global objective function, and we believe that our approach is a strong basis for future exploration/exploitation algorithms.

We studied online ranking in a stationary environment. Several works consider multi-step recommendations scenarios with dynamic models of content production \citep{mladenov2020optimizing,zhan2021towards}. They study the effect of including an exposure objective on the long-term user utilities, but they do not focus on how to efficiently generate rankings.

\paragraph{Relationship to Frank-Wolfe algorithms} The problem of inferring ranking lies in between convex bandit optimization \citep[see ][and references therein]{berthet2017fast} and stochastic optimization. Our problem is easier than bandit optimization since the function is known -- at least partially, and in all cases there is no need for active exploration. The main ingredient we add to the convex bandit optimization literature is the multi-user structure, where parameters are decomposed into several blocks that can only be updated one at a time, while optimizing for a non-decomposable objective function. The similarity with the bandit optimization algorithm of \citet{berthet2017fast} is the usage of the Frank-Wolfe algorithm to generate a deterministic decision at each step while implicitly optimizing in the space of probability distributions.

Our algorithm is a Frank-Wolfe algorithm with a stochastic gradient \citep{hazan2012projection,lafond2015online} and block-separable constraints \citep{lacoste2013block,kerdreux2018frank}. The difference with this line of work is twofold.
First, the distribution $\w$ is not necessarily uniform. Second, in our case, different users have different ``stepsizes'' for their parameters (the stepsize is $\frac{1}{\cti}$ for the user $i$ sampled at time $t$), rather than a single predefined stepsize. These two aspects complicate the analysis compared to that of the stochastic Frank-Wolfe with block-separable constraints of \citet{lacoste2013block}.
}
{
\section{Conclusion and discussion}\label{sec:discussion}
We presented a general approach to online ranking by optimizing trade-offs between user performance and fairness of exposure. The approach only assumes the objective function is concave and smooth. We provided three example tasks involving fairness of exposure, and the scope of the algorithm is more general. For instance, it also applies to the formulation of \citet{do2021two} for reciprocal recommendation tasks such as dating applications.

Despite the generality of the framework, there are a few technical limitations that could be addressed in future work. First, the assumption of the position-based model \eqref{eq:position_based} is important in the current algorithmic approach, because it yields the linear structure with respect to exposure that is required in our Frank-Wolfe approach. Dealing with more general cascade models \citep{craswell2008experimental,mansoury2021unbiased} is an interesting open problem. Second, we focused on the problem of generating rankings, assuming that (user, item) values $\muij$ are given by an oracle and are stationary over time. Relatedly to this stationarity assumption, we ignored the feedback loops involved in recommendation. These include feedback loops due to learning from prior recommendations \citep{bottou2013counterfactual}, the impact of the recommender system on users' preferences themselves \citep{kalimeris2021preference}, as well as the impact that fairness interventions on content production \citep{mladenov2020optimizing}. Third, our approach to balanced exposure is based on the knowledge of a discrete sensitive attribute of users. Consequently, this criterion cannot be applied when there are constraints on the direct usage of the sensitive attribute within the recommender system, when the sensitive attribute is not available, or when the delineation of groups into discrete categories is not practical or ethical \citep{tomasev2021fairness}. 

Finally, while we believe fairness of exposure is an important aspect of fairness in recommender systems, it is by no means the only one. For instance, the alignment of the system's objective with human values \citep{stray2021you} critically depends on the definition of the quality of a recommendation, the values $\muij$ in our framework. The fairness of the underlying system relies on careful definitions of these $\muij$s and on unbiased estimations of them from user interactions -- in particular, taking into account the non-stationarities and feedback loops mentioned previously.
}

\section*{Acknowledgements}
The authors thank Alessandro Lazaric and David Lopez-Paz for their feedback on the paper.

\bibliographystyle{unsrtnat}
\bibliography{references}

%%
%% If your work has an appendix, this is the place to put it.
\appendix
{
\section{Proof of Theorem \ref{thm:boundgeneral}}
\label{sec:proof:thm:boundgeneral}

In this section, we prove Theorem \ref{thm:boundgeneral}.

\subsection{Preliminary remarks}

We make here some preliminary observations that we take for granted in the proof. These are well-known in the analysis of variants of Frank-Wolfe.

We start with an observation that is crucial to the analysis of Frank-Wolfe algorithms, which is the following direct consequence of the concavity and differentiability of $\obj$:
\begin{align}\label{eq:regsmallerthan}
    \objwmax-\objw \leq \max_{\bar{\a}\in\simpA^n} \dotp{\nabla \objw(\pi)}{\bar{\a}}
\end{align}

The second observation is that the block-separable structure in Frank-Wolfe allows to solve for each user independently \citep[also see][Eq. 16]{lacoste2013block}:
\begin{align}\label{eq:decomposemax}
    \forall\pi\in\simpA^n,~~ \max_{\bar{\a}\in\simpA^n} \dotp{\nabla \objw(\pi)}{\bar{\a}} = \sum_{i=1}^n \max_{\bar{\a}\in\simpA} \dotp{\nablai \objw(\pi)}{\bar{\a}} = \sum_{i=1}^n \max_{\a\in\arms} \dotp{\nablai \objw(\pi)}{\a}.
\end{align}
These equalities are straightforward, and the last inequality comes from the fact that for a linear program solved over a polytope, there exists an extreme point of the polytope that is optimal.

{%%% block for temporary macros

The second observation relates to the use of approximate gradients \citep{lafond2015online,berthet2017fast}. Recall that $\boundA = \max_{\a\in\arms} \norm{\a}_1$ is the maximum $1$-norm of arms. Let $\g, \gh\in\Re^m$ and let \begin{align*}
    \aopt&
    \in\argmax_{\a\in\arms} \dotp{\g}{\a}& 
    \ha&
    \in\argmax_{\a\in\arms} \dotp{\gh}{\a}
\end{align*}
Then
\begin{align}\label{eq:approxgrad}
    \dotp{\g}{\ha} & = \dotp{g}{\aopt} + \dotp{\g}{\ha-\aopt} =
    \dotp{\g}{\aopt} + \underbrace{\dotp{\gh}{\ha-\aopt}}_{\geq 0} + \dotp{\g-\gh}{\ha-\aopt}
     \geq \dotp{\g}{\aopt} - 2\boundA\norm{\g-\gh}_\infty
\end{align}
%%% end of temporary macros
}

\subsection{Proof of Thm~\ref{thm:boundgeneral}}

In the remainder, we denote by $\expect_t[X]$ the conditional expectation $\expect[X|i_1, ..., \it]$.

Let $t\geq 1$. to simplify notation, we use the following shorcuts:
\begin{align}
    \git&=\frac{1}{\wi} \frac{\partial \objw(\pihtmo)}{\partial \pii},
    & \aoptit &\in\argmax_{\a\in\arms} \dotp{\git}{\a}, & 
    \ahit &\in\argmax_{\a\in\arms} \dotp{\ghit}{\a}.
\end{align}

We also denote by $\objwpii(\pii')$ the partial function with respect to $\pii$:
\begin{align}
    \objwpii(\pii') = \objw(\pione, \ldots, \piimo, \pii', \ldots, \piipo, \ldots, \pin).
\end{align}

We now start the proof.
First, let us fix $\itpo$ and notice that $\pihtpo$ is such that:
\begin{align}
    \pihitpo = \begin{cases}
    \pihit&\text{~if~} i\neq\itpo\\
    \pihit + \frac{1}{\cti+1}(\ahitpo - \pihit)&\text{~if~} i=\itpo\\
    \end{cases}
\end{align}

%% begin temp macros
{

Thus, only $\pihitpot$ changes. 
Let $\curvif = \frac{1}{2}\lipgwi \max_{\a,\a'\in\arms}\norm{\a-\a'}_2^2 \leq \lipgwi\boundA$ because we assumed $\forall j, 0\leq \a_j\leq 1$.\footnote{In more details: $\forall j, 0\leq \a_j\leq 1$ implies $\norm{\a-\a'}_2^2 = \sum_{j=1}^m (\a_j-\a'_j)^2 \leq \sum_{j=1}^m |\a_j-\a'_j| \leq 2\max_{\a\in\arms}\norm{\a}_1=2\boundA$.}

By the concavity of $\objwpihtitpo$ and its Lipschitz continuous gradients, we have \citep[see e.g.][Sec. 4.1]{bottou2018optimization}:
\begin{align}
    \objwpihtitpo\big(\pihitpotpo\big) \geq \objwpihtitpo\big(\pihitpot\big) 
    + \dotp[\Big]{\w_{\itpo}\gitpotpo}{\frac{1}{\ctitpo+1}\big(\ahitpotpo - \pihitpot\big)}
    - \frac{\w_{\itpo}}{\big(\ctitpo+1\big)^2}\curvitpof
\end{align}

Let $\regt = \objwmax - \objw(\pihtpo)$. Noticing that given $i_1, ..., \itpo$, we have $\objwpihtitpo\big(\pihitpotpo\big) = \objw(\pihtpo)$, we can rewrite 
\begin{align}
    \regtpo \leq \regt
    - \dotp[\Big]{\gitpotpo}{\frac{\w_{\itpo}}{\ctitpo+1}\big(\ahitpotpo - \pihitpot\big)}
    + \frac{\w_{\itpo}}{\big(\ctitpo+1\big)^2}\curvitpof
\end{align}

}
%% end temp macros

Taking the expectation over $\itpo$ (still conditional to $i_1, ..., \it$) gives:
\begin{align}\label{eq:condexpreg}
    \expect_t[\regtpo] \leq & \regt 
    \underbrace{ - \sum_{i=1}^n\dotp[\Big]{\w_i\gitpo}{\frac{\w_i}{\cti+1}\big(\ahitpo - \pihit\big)}}_{\calRt}
    + \sum_{i=1}^n\frac{\w_i^2}{\big(\cti+1\big)^2}\curvif
\end{align}

\subsubsection{Step 1: stepsize} 
Let $\hgammaitpo = \frac{\w_i}{\cti+1}$ and $\gammatpo = \frac{1}{t+1}$. Using $\hgammaitpo = \gammatpo + (\hgammaitpo - \gammatpo)$, $\norm{\ahitpo - \pihit}_1 \leq 2\boundA$ and $\norm{\gitpo}_\infty \leq \boundigw$ in $\calRt$, we have:
\begin{align}\label{eq:approxgammastep}
\calRt \leq - \gammatpo\sum_{i=1}^n\dotp[\Big]{\w_i\gitpo}{\ahitpo - \pihit} 
 + 2\boundA\underbrace{\sum_{i=1}^n \w_i\big| \hgammaitpo - \gammatpo\big| \boundigw}_{\calGt}
\end{align}

\paragraph{Step 2: approximate gradients} 
Now using the fact that $\ahitpo$ maximizes the dot product with the approximate gradient $\ghitpo$ and using the inequality \eqref{eq:approxgrad} in \eqref{eq:approxgammastep}, we obtain:
\begin{align}\label{eq:approxgradstep}
\calRt \leq\underbrace{ - \gammatpo\sum_{i=1}^n\dotp[\Big]{\w_i\gitpo}{\aoptitpo - \pihit}}_{\leq - \gammatpo\regt \text{ ~~by \eqref{eq:regsmallerthan} and \eqref{eq:decomposemax}}} + 2\gammatpo\boundA\underbrace{\sum_{i=1}^n\w_i\norm{\gitpo-\ghitpo}_\infty}_{\calDt} + 2\boundA\calGt
\end{align}
Plugging into \eqref{eq:condexpreg}, we finally obtain (recall $\gammatpo = \frac{1}{t+1}$):
\begin{align}
    \expect_t[\regtpo] \leq & (1-\gammatpo)\regt 
    + 2\gammatpo\boundA\calDt+2\boundA\calGt+\gammatpo\underbrace{\sum_{i=1}^n\frac{\w_i^2(t+1)}{\big(\cti+1\big)^2}\curvif}_{\calCt}
\end{align}
Taking the full expectation over $i^{(1)}, ..., \it$, using the notation $\pregt=\expect[\regtpo]$ and dividing by $\gammatpo$ yields:
\begin{align}
    (t+1)\pregtpo \leq & t\pregt 
    + \expect\big[\calCt+2\boundA\calDt+2(t+1)\boundA\calGt\big]
\end{align}

We separately bound each term in order in Lemmas \ref{lem:Ct}, \ref{lem:Dt} and \ref{lem:Gt}, we obtain:
\begin{align}
    t\pregt &\leq \sum_{\tp=0}^{t-1}\Big(2\frac{\sum_{i=1}^n \curvif}{\tau+1} + \frac{3\boundA\sum_{i=1}^n \sqrt{\w_i}\boundigw}{\sqrt{\tp+1}}+\frac{2\boundA\sum_{i=1}^n \boundigw}{\tp+1}\Big) + 4\boundA\sqrt{t}\sum_{i=1}^n\w_i\boundegi\\ 
    &\leq 2\sum_{i=1}^n (\curvif+\boundA\boundigw)\ln(e t) + 6\boundA\sum_{i=1}^n (\boundigw+\boundegi)\sqrt{\wi t}
\end{align}
using $\sum_{\tp=1}^{t}\frac{1}{\sqrt{\tp}} \leq 2\sqrt{t}$ as in Lemma \ref{lem:Dt}.

\subsection{Technical lemmas}\label{sec:proofs:lemma}

These lemmas depend heavily on the following standard equality regarding a Binomial variable with success probability $p$ and $t$ trials \citep{chao1972negative}:
\begin{align}\label{eq:expectation_one_over_binomial}
    \expect_{X\sim{\rm Bin}(p, t)} \Big[\frac{1}{X+1}\Big] = \frac{1}{p(t+1)}(1-(1-p)^{t+1}) \leq \frac{1}{p(t+1)}.
\end{align}

\begin{lemma}\label{lem:Ct}
\begin{align}
    \expect[\calCt] = \sum_{i=1}^n\curvif\expect\Big[\frac{\w_i^2(t+1)}{\big(\cti+1\big)^2}\Big] \leq 2\frac{\sum_{i=1}^n \curvif}{t+1}
\end{align}
\end{lemma}
\begin{proof}
We use the following result, where ${\rm Bin}(w_i, t)$ denotes the binomial distribution with probability of success $w_i>0$ and $t$ trials:
\begin{align}
    \expect_{X\sim{\rm Bin}(w_i, t)} \Big[\frac{1}{(X+1)^2}\Big] \leq \frac{2}{w_i^2(t+1)^2}
\end{align}
The proof uses $\frac{1}{(X+1)^2} \leq \frac{2}{(X+1)(X+2)}$ and the result is obtained by direct computation of the expectation. The arguments are the same as those required to obtain the standard equality \eqref{eq:expectation_one_over_binomial}.

The result follows by noticing that $\cti\sim {\rm Bin}(\w_i, t)$, using the inequality above if $\w_i\neq 0$ and noticing that the inequality $\frac{\w_i^2(t+1)}{(\cti+1)^2} \leq \frac{2}{t+1}$ holds when $\w_i=0$.
\end{proof}

\begin{lemma}\label{lem:Dt}
{

Using \eqref{eq:approximategradient}, we have:
\begin{align}
    \sum_{\tp=0}^{t-1} \expect[\calDt] = \sum_{i=1}^n\w_i \sum_{\tp=1}^t \expect\Big[\norm{\gitp-\ghitp}_\infty\Big] \leq  2\sqrt{t}\sum_{i=1}^n\w_i\boundegi
\end{align}
}
\end{lemma}
\begin{proof}
It is straightforward to show by induction that $\sum_{\tp=1}^t \frac{1}{\sqrt{\tp}}\leq 2\sqrt{t}$ and recall that $\sum_{i=1}^n\w_i=1$.
\end{proof}

\begin{lemma}\label{lem:Gt}
\begin{align}
    (t+1)\expect[\calGt] = \sum_{i=1}^n \w_i\boundigw\expect\Big[ \Big|\frac{\w_i(t+1)}{\cti+1} - 1\Big|\Big] \leq \frac{3 \sum_{i=1}^n \sqrt{\w_i}\boundigw}{2\sqrt{t+1}}+\frac{ \sum_{i=1}^n\boundigw}{t+1}
\end{align}
\end{lemma}
\begin{proof}
Let $i\in\intint{n}$ and $t>0$ we have
\begin{align}\label{eq:bound_on_one_over}
    \expect\Big[ \Big|\frac{\w_i(t+1)}{\cti+1} - 1\Big|\Big] 
    \leq \underbrace{\expect\Big[\frac{|\w_i t-\cti|}{\cti+1}\Big]}_{\caltGt} + \underbrace{\expect{\Big[\frac{|1-\w_i|}{\cti+1}}\Big]}_{\leq \frac{1}{\w_i(t+1)} \text{ because } \cti\sim{\rm Bin}(\w_i, t)}
\end{align}
Focusing on $\caltGt$ and writing it as:
\begin{align}
    \caltGt = \expect\Big[\underbrace{\frac{|(\w_i+\cti) - \w_i (t+1)|}{\w_i (t+1)}}_{= a}\times\underbrace{\frac{\w_i (t+1)}{\cti+1}}_{=b}\Big]
\end{align}
using $ab\leq \frac{1}{2} \lambda a^2+\frac{1}{2} \frac{b^2}{\lambda}$ with $\lambda=\sqrt{\wi}\sqrt{t+1}$ we obtain:
\begin{align}
    \caltGt \leq \frac{\sqrt{t+1}}{2\w_i^{3/2}} \underbrace{\expect\Big[\Big(\frac{\w_i+\cti}{t+1}-\w_i\Big)^2\Big]}_{\leq\frac{\w_i(1-\w_i)}{t+1} \text{(variance of sums of independent r.v.s)}}
    + \frac{1}{2\sqrt{\wi}\sqrt{t+1}} \underbrace{\expect\Big[\frac{(\w_i (t+1))^2}{\big(\cti+1\big)^2}\Big]}_{\leq 2 \text{ by Lemma \ref{lem:Ct}}}.
\end{align}
And the result follows.
\end{proof}

}
{
\section{Proof of Proposition \ref{prop:approx_gradients}}
\label{sec:proof:approx_gradients}

We give here the computations that lead to the bounds of \ref{prop:approx_gradients}. The bounds themselves mostly rely on calculations of bounds on second-order derivatives, but they rely on a set of preliminary results regarding the deviation of multinomial distributions. We first give these results, and then detail the calculations for each fairness objective in separate subsections.

\textbf{In this section, we drop the superscript in $t$. All quantities with a hat $\hat{.}$ are implicitly taken at time step $t$, i.e. they should be read $\hat{.}^{(t)}$. This is also the case for $\pih$, which should read $\piht$.} We also remind that $\muij\in[0,1]$ so user utilities and item qualities are also in $[0,1]$, and $\norm{\w}_\infty \leq 1$ so all exposures are also in $[0,1]$.

\subsection{Main inequalities}
All our bounds are based on the following result \citep[Thm. 1]{han2015minimax}:
\begin{align}\label{eq:expectation_w_norm}
    \Big(\multiliner{8em}{1-norm distance between $\wh$ and $\w$}\Big)&&\expect\Big[\norm[\big]{\wh-\w}_1\Big] \leq \sqrt{\frac{n-1}{t}},
\end{align}
where the expectation is taken over the random draws of $\ione, \ldots, \it$.

From \eqref{eq:expectation_w_norm} we obtain a bound on the deviation of online estimates of exposures from their true value:
\begin{align}\label{eq:deviation_exposure}
\expect\Big[\big|\vhj - \vj(\pih)\big|\Big]  
&\leq 
\expect\Big[\sum_{i=1}^n\big| \whi-\wi\big| \pihij \Big] 
\leq \expect\Big[\norm[\big]{\wh-\w}_1\Big] \leq \sqrt{\frac{n-1}{t}}
\end{align}
For balanced exposure, we need a refined version of \eqref{eq:expectation_w_norm} regarding the convergence of per-group user activities:
\begin{lemma} For every group $\group$
\begin{align}
\expect\bigg[\Big|\sum_{i\in\group} \frac{\ci}{\cg} - \frac{\wi}{\wsumg}\Big|\bigg] \leq \sqrt{\frac{2(\card{\group}-1)}{\wsumg t}}
\end{align}
\end{lemma}
\begin{proof}
The lemma is a consequence of \eqref{eq:expectation_w_norm}, by first conditioning on the groups sampled at each round. Let $\groupone, \ldots, \groupt$ correspond to the sequence of groups sampled at each round. Since given the group, the sampled users are i.i.d. within that group, we can use \eqref{eq:expectation_w_norm}:
\begin{align}
    \expect\bigg[\Big|\sum_{i\in\group} \frac{\ci}{\cg} - \frac{\wi}{\wsumg}\Big|\bigg  | \groupone, \ldots, \groupt\bigg] \leq \indic{\cg>0}\sqrt{\frac{\card{\group}-1}{\cg}} + \indic{\cg=0}.
\end{align}
Notice that in the equation above we have $\ctg = \sum_{\tp=1}^t \indic{\grouptp = \group}$.

Using $\frac{\card{\group}-1}{\cg} \leq 2\frac{\card{\group}-1}{\cg+1}$ when $\cg>0$ and $\sqrt{2}>1$ when $\cg=0$, taking the expectation over the random draws of groups, we obtain:
\begin{align}
    \expect\bigg[\Big|\sum_{i\in\group} \frac{\ci}{\cg} - \frac{\wi}{\wsumg}\Big|\bigg] \leq \expect\Big[\sqrt{2\frac{\card{\group}-1}{\cg+1}}\Big]
    \leq 
    \sqrt{2\expect\Big[\frac{\card{\group}-1}{\cg+1}\Big]}
\end{align}
by Jensen's inequality. The result follows from the expectation of $\frac{1}{X+1}$ when $X$ follows a binomial distribution \eqref{eq:expectation_one_over_binomial}.
\end{proof}
The lemma above allows us to extend \eqref{eq:deviation_exposure} to exposures within a group:
\begin{align}\label{eq:deviation_exposure_withingroup}
    \expect\Big[\big|\vhgj - \vgj(\pih)\big|\Big]  \leq \sqrt{\frac{2(\card{\group}-1)}{\wsumg t}} 
    && 
    \text{and~~}
    \expect\Big[\big|\vhavgj - \vavgj(\pih)\big|\Big]  \leq \frac{1}{\card{\groups}}\sum_{\group\in\groups}\sqrt{\frac{2(\card{\group}-1)}{\wsumg t}} 
\end{align}

\subsection{Two-sided fairness (Alg.~\ref{alg:twosided})}
We have
$\pdevwij\big(\pih\big) = \psi'_{\alpha_1}\big(\ui(\pih\big)\mui + \frac{\beta}{m}\psi'_{\alpha_2}\big(\vj(\pih)\big)$, and, by definition, $\uhi=\ui(\pih)$ and $\pdevwhij(\pih) =\psi'_{\alpha_1}(\uhi)\muij + \frac{\beta}{m}\psi'_{\alpha_2}(\vhj)$.
We thus have:
\begin{align}
    \Big|\pdevwij\big(\pih\big) -\pdevwhtij(\pih)\Big|
    &=
    \frac{\beta}{m}\Big |\psi'_{\alpha_2}\big(\vj(\pih)\big) - \psi'_{\alpha_2}(\vhj)\Big | 
    \leq \frac{\beta\norm{\psi''_{\alpha_2}}_\infty}{m} 
    \Big|\vj(\pih)- \vhj\Big| \\
\end{align}
Taking the expectation over $\ione, \ldots,\it$ and using \eqref{eq:deviation_exposure}, we obtain the desired result:
\begin{align}
    \expect\Big[\Big|\pdevwij\big(\pih\big) -\pdevwhij(\pih)\Big|\Big] \leq \frac{\beta\norm{\psi''_{\alpha_2}}_\infty}{m}\sqrt{\frac{n-1}{t}}
\end{align}

\subsection{Quality-weighted exposure (Alg.~\ref{alg:quaexpo})} By similar direct calculations, let
\begin{align}
    \displaystyle \hat{Z} = \sqrt{\eta+\frac{1}{m} \sum_{j=1}^m \Big(\quahavg\vhj-\quahj\norm{\wei}_1\Big)^2}
&&\text{and~~} \displaystyle \hat{Z} = \sqrt{\eta+\frac{1}{m} \sum_{j=1}^m \Big(\quaavg\vj(\pih)-\quaj\norm{\wei}_1\Big)^2}.
\end{align}
\begin{align}\label{eq:starter_qua_approx_gradient}
    \Big|\pdevwij\big(\piht\big) -\pdevwhij(\pih)\Big| 
    &= \frac{\beta}{m}\Big|\frac{\quahavg\vhj - \quahj\norm{\wei}_1}{\hat{Z}}- \frac{\quaavg\vj(\pih) - \quaj\norm{\wei}_1}{Z}\Big|
\end{align}
Notice first that given some $B>0$, for $x\in[-B,B]^m$, the function $h(x) =\sqrt{\eta+\frac{1}{m}\sum_{j=1}^m x_j^2}$ has derivatives bounded by $\frac{B}{m\sqrt{\eta}}$ in $\norm{.}_\infty$. Thus, we have, for every $x, x'\in[-B,B]^m$:
\begin{align}\label{eq:main_simplification}
    \Big|\sqrt{\eta+\frac{1}{m}\sum_{j=1}^m x_j^2} - \sqrt{\eta+\frac{1}{m}\sum_{j=1}^m {x'_j}^2}\Big| \leq \frac{B}{m\sqrt{\eta}}\sum_{j=1}^m |x_j-x'_j|
\end{align}
With $x_j = \quaavg\vj(\pih)-\quaj\norm{\wei}_1$ and $x'_j = \quahavg\vhj-\quahj\norm{\wei}_1$, we have $B\leq 1+\norm{\wei}_1$. Moreover, the bound \eqref{eq:starter_qua_approx_gradient} writes:
\begin{align}\label{eq:details_approx_grad}
    \frac{m}{\beta}\Big|\pdevwij\big(\pih\big) -\pdevwhij(\pih)\Big| 
    &= \Big| \frac{x_j}{h(x)}- \frac{x'_j}{h(x')}\Big| \leq \frac{|x_j-x'_j|h(x)}{h(x)h(x')}
    + |x'_j|\frac{\big| h(x) - h(x')\big|}{h(x)h(x')}\\
    & \leq \frac{|x_j-x'_j|}{\sqrt{\eta}} + \frac{B}{m\eta}\norm{x-x'}_1 
    \leq \frac{1+B}{m\min(\eta,\sqrt{\eta})}\norm{x-x'}_{\infty},
\end{align}
where we used \eqref{eq:main_simplification} and the fact that $h(x) \geq \sqrt{\eta}$ and $h(x') \geq |x'_j|$.
Now, since all exposures and qualities are upper bounded by $1$, notice that
\begin{align}
    |x_j - x'_j| \leq |\quaavg - \quahavg| + |\vhj-\vj(\pih)| + \norm{\wei}_1 |\quaj - \quahj| \leq (2+\norm{\wei}_1)\norm{\wh-\w}_1,
\end{align}
where we used similar calculations as in \eqref{eq:deviation_exposure} to bound for $\quaj - \quahj$ and $\quaavg - \quahavg$.

Putting it all together, and using \eqref{eq:deviation_exposure} for the expectation, we obtain:
\begin{align*}
    \expect\Big[\Big|\pdevwij\big(\pih\big) -\pdevwhij(\pih)\Big|\Big] 
    % \leq 
    % \frac{\beta}{m}\Big(\frac{2+\norm{\w}_1}{\sqrt{\eta}} + \frac{\big(2+\norm{\w}_1\big)^2}{\eta^{3/2}}\Big)\expect\Big[\norm{\wh-\w}_1\Big] 
    \leq 
    \frac{\beta\big(2+\norm{\wei}_1\big)^2}{m\min(\eta,\sqrt{\eta})}\sqrt{\frac{n-1}{t}}.
\end{align*}

\subsection{Balanced exposure} For balanced exposure, let us denote by
\begin{align}
    \hat{Z}_j = \sqrt{\eta + \sum_{\group\in\groups} \Big(\vhgj-\vhavgg\Big)^2}
    &&
    \text{and~~} Z_j = \sqrt{\eta + \sum_{\group\in\groups} \Big(\vgj(\pih)-\vavgj(\pih)\Big)^2}.
\end{align}
We then have\footnote{We notice that in Alg.~\ref{alg:balancedexpo} we have $\frac{t}{\cgi+1}$ rather than $\frac{t+1}{\cgi+1}$. This is because here we are considering $\pdevwhti\big(\piht\big) = \ghitpo$ (see \eqref{eq:which_approximate_gradient}), while the algorithm uses $\ghit$.}
\begin{align}
    \Big|\pdevwij\big(\pih\big) -\pdevwhij(\pih)\Big| 
    &= \frac{\beta}{m}\Big|\frac{t+1}{\cgi+1}\frac{\vhgij - \vhavgj}{ \hat{Z}_j} - \frac{1}{\wsumgofi}\frac{\vgij(\pih) - \vavgj(\pih)}{Z_j} \Big|.
\end{align}
Similarly to quality of exposure, for $x\in[-B,B]^{\card{\groups}}$ let us denote by $h(x) = \sqrt{\eta + \sum_{\group\in\groups} \xg^2}$ for $\xg\in[0,B]^{\card{\groups}}$. We have $|h(x) - h(x')| \leq \frac{B}{\sqrt{\eta}}\sum_{\group\in\groups} |x_\group-x'_\group|$. Moreover, with 
\begin{align}
    \xg =\vhgj - \vhavgj && \xg' = \vgj(\pih) - \vavgj(\pih) && \alpha_j=\frac{t+1}{\cgi+1} && \alpha_j' = \frac{1}{\wsumgofi}
\end{align}
we can use $B=1$, and, using similar steps as \eqref{eq:details_approx_grad} (here the gradients of $h$ are bounded by $\frac{B}{\sqrt{\eta}}$ in infinity norm):
\begin{align}
    \Big|\pdevwij\big(\pih\big) -\pdevwhij(\pih)\Big| 
    &= \frac{\beta}{m} \bigg(\frac{B}{\sqrt{\eta}} |\alpha_i-\alpha'_i| +\alpha'_i\Big|\frac{\xgi}{h(x)} -\frac{\xgi'}{h(x')} \Big|\bigg)\\
    &\leq
    \frac{\beta}{m\wsumgofi\min(\eta,\sqrt{\eta})} \bigg( \wsumgofi|\alpha_j-\alpha'_j| + |\xgi - \xgi'| + \norm{x-x'}_1\bigg)
\end{align}
We now notice that using \eqref{eq:deviation_exposure_withingroup}, we have:
\begin{align}\label{eq:balanced_bounds_onemorn}
    \expect\Big[|\xgi - \xgi'|\Big] \leq \sqrt{\frac{2(\card{\groupofi}-1)}{\wsumgofi t}} + \frac{1}{\card{\groups}}\sum_{\group\in\groups}\sqrt{\frac{2(\card{\group}-1)}{\wsumg t}} 
    && \text{and thus } 
    \expect\Big[\norm{x-x'}_1\Big]\leq 2\sum_{\group\in\groups}\sqrt{\frac{2(\card{\group}-1)}{\wsumg t}} 
\end{align}
We finish the proof with the bound on  $|\alpha_i-\alpha_i'|$:
\begin{align}
    \expect\big[|\alpha_i-\alpha'_i|\big]
    & 
    = \expect\bigg[
    \Big| \frac{t+1}{\cgi+1} - \frac{1}{\wsumgofi}\Big| \Bigg]
    = \frac{1}{\wsumgofi}\expect\bigg[\frac{\big|\wsumgofi + (\wsumgofi-1) \big|}{\cgi+1}\bigg]
\end{align}
Following the same steps as \eqref{eq:bound_on_one_over} in Lem.~\ref{lem:Gt}, we obtain:
\begin{align}
    \expect\bigg[\frac{\big|\wsumgofi + (\wsumgofi-1) \big|}{\cgi+1}\bigg]
    \leq \frac{3}{2\sqrt{\wsumgofi(t+1)}}+\frac{1}{\wsumgofi (t+1)}
\end{align}
Putting it all together, we get:
\begin{align}
   \mathbb{E} \left[ \Big|\pdevwij\big(\pih\big) -\pdevwhij(\pih)\Big| \right] \leq
    \frac{\beta}{m\wsumgofi \min(\eta,\sqrt{\eta})}\bigg(
    \frac{3}{2\sqrt{\wsumgofi(t+1)}} + \frac{1}{\wsumgofi(t+1)}
    +
    4\sum_{\group\in\groups}\sqrt{\frac{2(\card{\group} - 1)}{\wsumg t}}\bigg).
\end{align}
Finally, using $4\sqrt{2(\card{\group} - 1)} + \frac{3}{2} \leq 8\sqrt{\card{\group}}$ as long as $\card{\group}\geq 1$ (we assumed groups are non-empty),\footnote{The function $\group\mapsto 4\sqrt{2(\card{\group} - 1)} + \frac{3}{2} - 8\sqrt{\card{\group}}$ is decreasing and is $\leq 0$ when $\card{\group}=1$.} we obtain 
\begin{align}
   \mathbb{E} \left[ \Big|\pdevwij\big(\pih\big) -\pdevwhij(\pih)\Big| \right] \leq
    \frac{\beta}{m\wsumgofi\min(\eta,\sqrt{\eta})}\bigg(
    \frac{1}{\wsumgofi(t+1)}
    +
    8\sum_{\group\in\groups}\sqrt{\frac{\card{\group}}{\wsumg t}}\bigg).
\end{align}
which simplifies to the desired result when $\eta\leq 1$.

\section{Proof of Corollary \ref{cor:boundalgos}}\label{sec:proof:boundalgos}
When $\eta\leq 1$, and using $A\lesssim B$ as a shorthand for $A=O(B)$, one can bound the constants $D_i$ in Theorem \ref{thm:boundgeneral} as follows.
\begin{itemize}
    \item Alg.~\ref{alg:twosided} (two-sided fairness): Using Proposition \ref{prop:approx_gradients}, $D_i$ taken as follows is sufficient
    $$
    D_i \lesssim \dfrac{\beta\|\psi''_{\alpha_2}\|_\infty\sqrt{n}}{m} 
    \lesssim \dfrac{\beta\eta^{\alpha_2-2}\sqrt{n}}{m} 
    %\lesssim \dfrac{\beta\sqrt{n}}{m} =: C_1.
    $$
    \item Alg.~\ref{alg:quaexpo} (quality-weighted):
    $$
    D_i \lesssim \dfrac{\beta(2+\|b\|_1)^2\sqrt{n}}{m\min(\eta,\sqrt{\eta})}
    % \lesssim \dfrac{\|b\|_1^2\sqrt{n}}{m} 
    \lesssim \dfrac{\|b\|_1^2\beta\sqrt{n}}{m\eta}
    %=: C_2.
    $$
    
    \item Alg.~\ref{alg:balancedexpo} (balanced exposure):
    By euclidean Cauchy-Schwarz inequality, we have
    $$
    \left(\frac{1}{|\mathcal S|}\sum_{s \in \mathcal S} \sqrt{\frac{|s|}{\wsumg}}\right)^2 \le \left(\frac{1}{|\mathcal S|}\sum_s |s|\right)\left(\frac{1}{|\mathcal S|}\sum_s \frac{1}{\wsumg}\right) \le \frac{n}{|\mathcal S|\wming}
    $$
    where $\wming := \min_{s \in \mathcal S}\wsumg$. Thus, we deduce that $\sum_{s \in \mathcal S}\sqrt{|s|/\wsumg} \le \sqrt{n|\mathcal S|/\wming}$, and so
    $$
    D_i \lesssim \dfrac{\beta}{m\wsumgofi\min(\eta,\sqrt{\eta})} \sum_{s \in \mathcal S}\sqrt{\dfrac{|s|}{\wsumg}} \le \dfrac{\beta\sqrt{n|\mathcal S|/\wming}}{m\wsumgofi\min(\eta,\sqrt{\eta})}
    %\lesssim \frac{\beta}{\wsumgofi} \dfrac{\sqrt{n|\mathcal S|}}{m} \frac{1}{\sqrt{\wming}} \le  \dfrac{\beta\sqrt{n|\mathcal S|}}{m} \frac{1}{\wming^{3/2}} =
    %
    \lesssim\frac{\beta\sqrt{n}}{m\eta}\sqrt{\frac{|\mathcal S|}{\wming^3}} 
    %=:C_3.
    $$
\end{itemize}

Thus, the gradient estimates \eqref{eq:approximategradient} required in Theorem \ref{thm:boundgeneral} hold with $D_i= \mathcal O(C_\star)$ for all $i$, where the constants $C_\star$ are given above.

In addition, we can note that 
\begin{itemize}
    \item $\ln(et)/t \ll 1/\sqrt{t}$ and so the first term in \eqref{eq:generalregret} is dominated by the second, hence can be ignored.
    \item The $L_i$'s are bounded independently of $t$ and only impact the $\log(et)/t$ term, so we can ignore them.
    \item $\boundA = \norm{\wei}_1$.
    \item By Jensen's inequality, $(1/n)\sum_i \sqrt{w_i} \le \sqrt{(1/n)\sum_i w_i} = \sqrt{1/n}$, and so $\sum_i \sqrt{w_i} \le \sqrt{n}$.
\end{itemize}
The regret bounds presented in Table~\ref{tab:rates} then follow upon plugging these estimates for $D_i$ into the generic regret bound \eqref{eq:generalregret} of Theorem \ref{thm:boundgeneral}, together with the following bounds for $\boundigw$:
\begin{itemize}
\item Alg.~\ref{alg:twosided} (two-sided fairness) $\boundigw \lesssim \norm{\psi'_{\alpha_1}}_\infty+\frac{\beta}{m}\norm{\psi'_{\alpha_2}}_\infty \leq \eta^{\alpha_1-1} + \frac{\beta}{m}\eta^{\alpha_2-1}$,
\item Alg.~\ref{alg:quaexpo} (quality-weighted) $\boundigw \leq 1 + \frac{\beta}{m\sqrt{\eta}}$,
\item Alg.~\ref{alg:balancedexpo} (balanced exposure) $\boundigw \leq 1 + \frac{\beta}{m\wming\sqrt{\eta}}$.
\end{itemize}
Notice that the user activity does not appear explicitly in these upper-bounds on the normalized gradient because (1) the user objective weights $\pii$ by $\wi$ and (2) in the definition of item exposure, $\piij$ is also weighted by $\wi$. The normalization removes these weights.

}
{
\section{Additional experimental details}

\subsection{Training}\label{sec:xp_training}

For the \lastfm experiments, following \citet{do2021two}, the training was performed by randomly splitting the dataset into 3 splits $70\%/10\%/20\%$ of train/validation/test sets. The hyperparameters of the factorization\footnote{Using the Python library \texttt{Implicit} toolkit \url{https://github.com/benfred/implicit}.} are selected on the validation set by grid search. The number of latent factors is chosen in $[16, 32, 64, 128]$, the regularization in $[0.1, 1., 10., 20., 50.]$, and the confidence weighting parameter in $[0.1, 1., 10., 100.]$. The estimated preferences we use are the positive part of the resulting estimates.

\subsection{More details on \lastfm}\label{sec:xp:details:lastfm}

\begin{figure}[t]
    \centering
    \includegraphics[width=\linewidth]{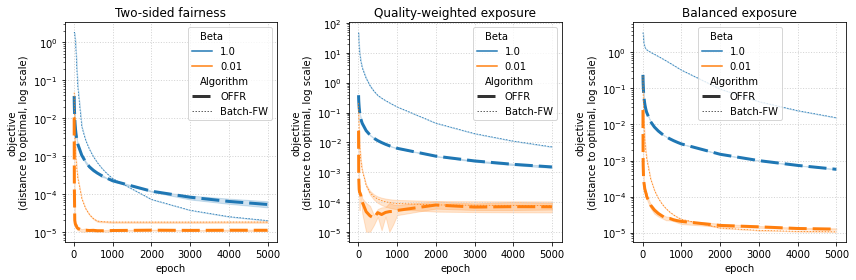}
    
    \includegraphics[width=\linewidth]{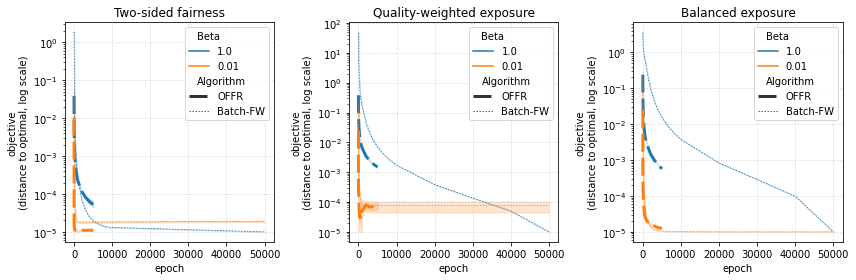}
    \caption{Convergence speed on \lastfm of \ofr compared to \batchfw on the three fairness objectives, for $\beta \in\{0.01, 1\}$ and $\eta=1$. As expected, they converge to the same values. \ofr was run for $5k$ epochs, while \batchfw was run for $50k$ epochs. We see that \ofr converges to the same objective function value as \batchfw as expected, up to some noise on quality-weighted exposure for small values of $\beta$.}
    \label{fig:convergence_details_lastfm}
\end{figure}

\begin{figure}[t]
    \centering
    \includegraphics[width=\linewidth]{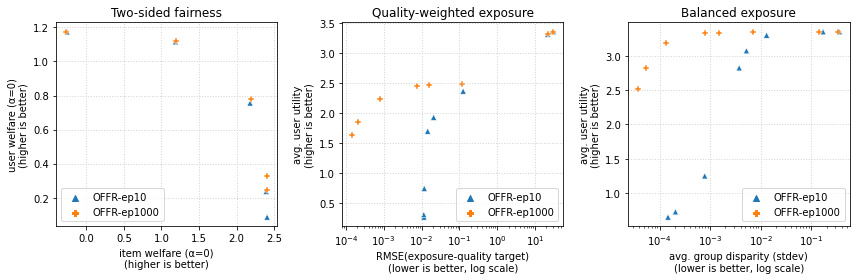}
    
    \includegraphics[width=\linewidth]{trade-offs_1.png}
    \caption{Comparison of the trade-offs obtained by varying $\beta$ for $\eta=0.01$ (top row) and $\eta=1$ (bottom row, repeating Fig.~\ref{fig:tradeoffs} for better visibility), for \ofr with $10$ and $1000$ epochs on \lastfm. }
    \label{fig:tradeoffs_details_lastfm}
\end{figure}

\begin{figure}[t]
    \centering
    \includegraphics[width=\linewidth]{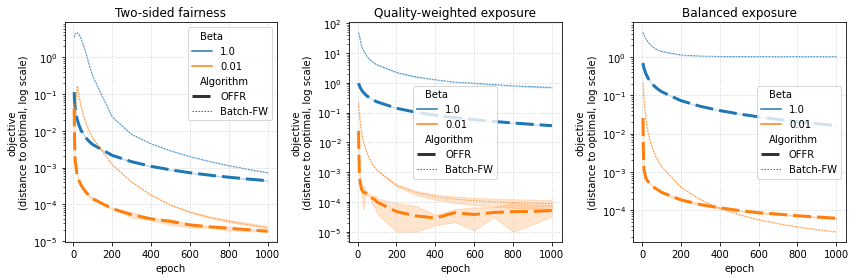}
    \caption{Convergence speed of \ofr compared to \batchfw on \lastfm  for $\eta=0.01$. The overall trends are similar as for $\eta=1$ in Fig.~\ref{fig:convergence}, except that \batchfw becomes more rapidly better than \ofr on the balanced exposure objective for small $\beta$.}
    \label{fig:convergence_eta_lastfm}
\end{figure}

In this section, we provide additional details regarding convergence, as well as the choice of $\eta$. 

\paragraph{Convergence} In Fig.~\ref{fig:convergence_details_lastfm} we give the results of the algorithms with more epochs than in Fig.~\ref{fig:convergence} in the main paper. The online algorithm was run for 5000 epochs, while the batch algorithm was run for $50k$ epochs (which was necessary to reach the convergence value for batch for large values of $\beta$). We observe of online on two-sided fairness and balanced exposure, the convergence on quality weighted exposure is more noisy and seems to oscillate around $10^{-4}/10^{-5}$ of the objective, but nonetheless converges to the desired value with much faster convergence as beta becomes large.

\paragraph{Changing $\eta$} In our objective functions, the main purpose of $\eta$ is to ensure that the objective functions are smooth. Note that fundamentally, we are looking from trade-offs between a user objective and a fairness objective by varying $\beta$. Different values of $\eta$ lead to different trade-offs for fixed $\beta$, but not necessarily different Pareto fronts when varying $\beta$ from very small to very large values. Nonetheless, as $\eta$ controls the curvature of the function, it is important for the convergence of both \batchfw (see e.g., the analysis in \citet{clarkson2010coresets} for more details in the convergence of batch Frank-Wolfe algorithms and the importance of the curvature). 

In Fig.~\ref{fig:tradeoffs_details_lastfm}, we show the trade-offs achieved with $\eta=0.01$ compared to $\eta=1$ as shown in the main paper, for the same values of $\beta\in\{10^{x}, x\in\{-3, -2, -1, 0, 1, 2\}\}$, as in Fig.~\ref{fig:tradeoffs}. 
% For two-sided fairness, we plot the values of the user and item objectives, which themselves depend on $\eta$. We give the plots for exhaustivity but they are not conclusive as they stand. Nonetheless, 
For both quality-weighted and balanced exposure, we observe that the smaller values of $\eta$ (top row) reaches better trade-offs than $\eta=1$ for this range of $\beta$, which may indicate that $\eta=0.01$ might be preferable in practice to $\eta=1$. 

In Fig.~\ref{fig:convergence_eta_lastfm}, we plot the convergence speed of \ofr and \batchfw for the first $1000$ epochs when $\eta=0.01$. comparing with Fig.~\ref{fig:convergence}, we observe that, as expected, both \ofr and \batchfw converge to their objective function slowlier overall. The relative convergence of \ofr compared to \batchfw follow similar trends as for $\eta=1$, with \ofr obtaining better values of the objective at the beginning, and \ofr being significantly better than \batchfw for large values of $\beta$. 

Interestingly though, compared to \batchfw, \ofr still converges relatively fast, and seems less affected by the larger curvature than \batchfw. This is coherent with the observation that \ofr converges faster than \batchfw at the beginning, especially for large values of $\beta$. While the exact reason why \ofr is less affected by large curvatures than \batchfw is an open problem, these results are promising considering the wide applicability of online Frank-Wolfe algortihms.

\subsection{Results on Movielens data}\label{sec:xp_mlm}

To show the results on another dataset, we replicate the experiments on the MovieLens-1m dataset (\mlm) \cite{harper2015movielens}. The dataset contains ratings of movies, with $~3k$ users and $4k$ items, as well as a gender attribute. We transform the rating matrix as a binary problem by considering ratings $\geq 3$ as positive examples, and setting all other entries as 0. This makes the problem similar to \lastfm, and we then follow exactly the same as for \lastfm, including sampling and learning the matrix factorization of user values.

\begin{figure}[t]
    \centering
    \includegraphics[width=\linewidth]{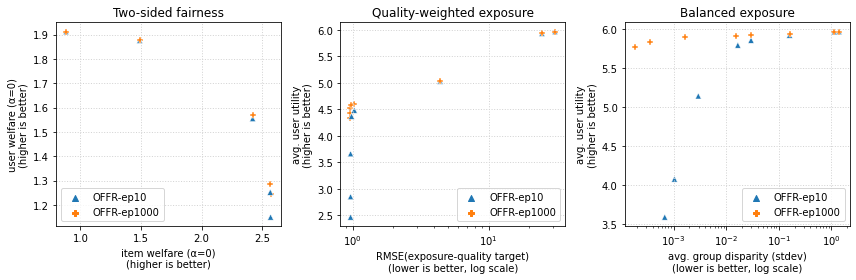}
    \caption{Trade-offs in terms of user objective ($y$-axis) and item fairness ($x$-axis) for \mlm. The observations are similar than on \lastfm.}
    \label{fig:tradeoffs_ml1m}
\end{figure}

\begin{figure}[t]
    \centering
    \includegraphics[width=\linewidth]{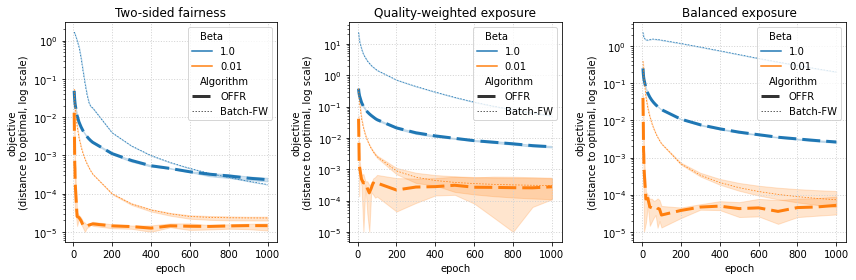}
    
    \includegraphics[width=\linewidth]{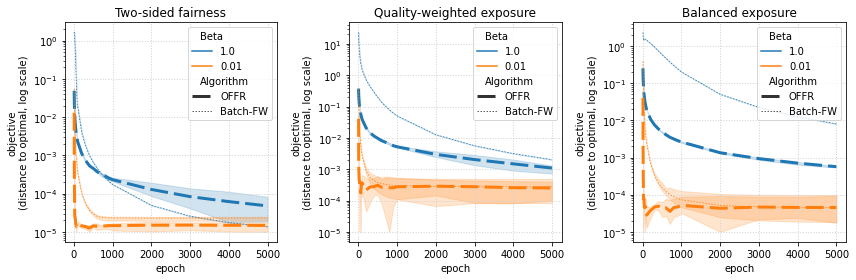}
    
    \includegraphics[width=\linewidth]{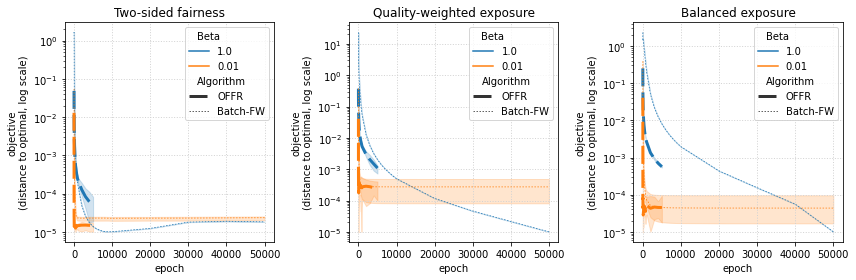}
    \caption{Convergence speed of \ofr compared to \batchfw on \mlm,  for $\beta \in\{0.01, 1\}$ and $\eta=1$ for the first $1k$ epochs (top row), $5k$ epochs (middle row) and $50k$ epochs (bottom row, note that only \batchfw was ran for $50k$ epochs.}
    \label{fig:convergence_ml1m}
\end{figure}

\begin{figure}[t]
    \centering
    \includegraphics[width=0.48\linewidth]{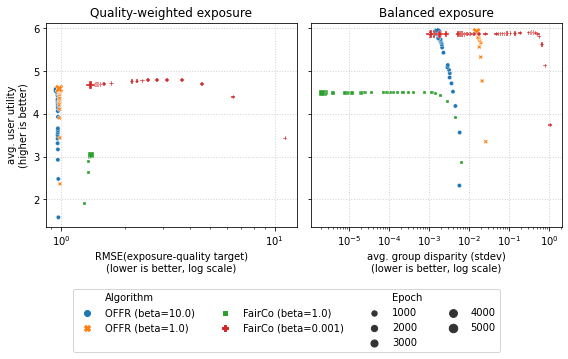}
    \hfill
    \includegraphics[width=0.48\linewidth]{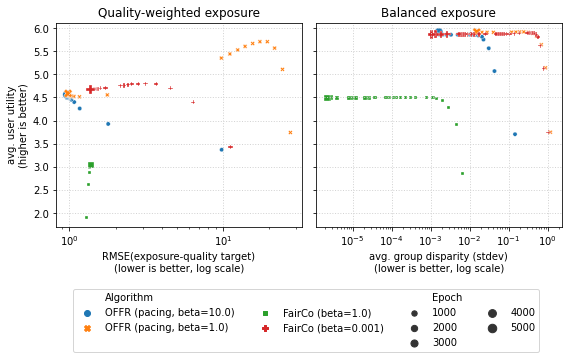}
    \caption{Convergence of \ofr compared to the dynamics of \fairco on \mlm (left) without the pacing heuristic (right) with the pacing heuristic.}
    \label{fig:fairco_ml1m}
\end{figure}

\paragraph{Qualitative trade-offs} The trade-offs obtained by varying $\beta\in\{10^{x}, x\in\{-3, -2, -1, 0, 1, 2\}\}$ are shown in Fig.~\ref{fig:tradeoffs_ml1m}. The plots are qualitatively very similar to those of
\lastfm.

\paragraph{convergence} The convergence of \ofr compared to \batchfw are shown in Fig.~\ref{fig:convergence_ml1m}. They again look very similar to the plots of \lastfm, with \ofr being better than \batchfw at the beginning, and reaching the same values at convergence than \batchfw, up to noise when the objective is close to the optimum.

\paragraph{Comparison to \fairco} Fig.~\ref{fig:fairco_ml1m} shows the comparison to \fairco. The observations are once again similar to those on \lastfm, with the main trends exacerbated (\ofr without pacing heuristic converging in fairness objective extremely fast and taking time to converge in user objective, while \fairco keeping a high value of user objective as much as possible while consistently reducing the item objective.  
}
\end{document}